\documentclass[reqno,10pt,letterpaper]{amsart}

\usepackage{amsmath,amssymb,amsthm,mathrsfs,bbm}
\usepackage{accents}
\usepackage{arydshln}
\usepackage{upgreek}
\usepackage{slashed}
\usepackage{xifthen}
\usepackage{graphicx}
\usepackage{longtable}
\usepackage[inline]{enumitem}

\usepackage{xcolor}
\definecolor{winered}{rgb}{0.7,0,0}
\definecolor{lessblue}{rgb}{0,0,0.7}

\usepackage[pdftex,colorlinks=true,linkcolor=winered,citecolor=lessblue,breaklinks=true,bookmarksopen=true]{hyperref}

\makeatletter
\newcommand{\myitem}[3]{\item[#2]\def\@currentlabel{#3}\label{#1}}
\makeatother

\usepackage{titletoc}

\makeatletter

\def\@tocline#1#2#3#4#5#6#7{
\begingroup
  \par
    \parindent\z@ \leftskip#3 \relax \advance\leftskip\@tempdima\relax
                  \rightskip\@pnumwidth plus 4em \parfillskip-\@pnumwidth
    \ifcase #1 
       \vskip 0.6em \hskip 0em 
       \or
       \or \hskip 0em 
       \or \hskip 1em 
    \fi%
    #6
    \nobreak\relax{\leavevmode\leaders\hbox{\,.}\hfill}
    \hbox to\@pnumwidth {\@tocpagenum{#7}}
  \par
\endgroup
}

 \def\l@section{\@tocline{0}{0pt}{0pc}{}{}}

\renewcommand{\tocsection}[3]{%
  \indentlabel{\@ifnotempty{#2}{ 
    \ignorespaces\bfseries{#2. #3}}}
  \indentlabel{\@ifempty{#2}{\ignorespaces\bfseries{#3}}{}} 
    \vspace{1.5pt}}

\renewcommand{\tocsubsection}[3]{%
  \indentlabel{\@ifnotempty{#2}{
    \ignorespaces#2. #3}}
  \indentlabel{\@ifempty{#2}{\ignorespaces #3}{}}
    \vspace{1.5pt}}

\renewcommand{\tocsubsubsection}[3]{%
  \indentlabel{\@ifnotempty{#2}{
    \ignorespaces#2. #3}}
  \indentlabel{\@ifempty{#2}{\ignorespaces #3}{}}
    \vspace{1.5pt}}

\makeatother

\makeatletter
\def\@nomenstarted{0}
\newlength{\@nomenoldtabcolsep}

\newcommand{\nomenstart}
  {%
    \def\@nomenstarted{1}%
    \setlength{\@nomenoldtabcolsep}{\tabcolsep}%
    \setlength{\tabcolsep}{3.5pt}%
    \begin{longtable}{p{0.11\textwidth} p{0.86\textwidth}}
  }

\newcommand{\nomenitem}[2]{%
    \ifcase\@nomenstarted%
      \or 
      \or \\ 
    \fi%
    #1\,{\leavevmode\leaders\hbox{\,.}\hfill} & #2%
    \def\@nomenstarted{2}%
  }%
\newcommand{\nomenend}
  {\\%
      \end{longtable}%
      \setlength{\tabcolsep}{\@nomenoldtabcolsep}%
      \def\@nomenstarted{0}%
  }
\makeatother

\makeatletter
\newcommand{\vast}{\bBigg@{4}}
\newcommand{\Vast}{\bBigg@{5}}
\makeatother

\allowdisplaybreaks

\numberwithin{equation}{section}
\numberwithin{figure}{section}
\newtheorem{thm}{Theorem}[section]

\newtheorem{prop}[thm]{Proposition}
\newtheorem{lemma}[thm]{Lemma}

\newtheorem*{thm*}{Theorem}
\newtheorem*{prop*}{Proposition}
\newtheorem*{cor*}{Corollary}
\newtheorem*{conj*}{Conjecture}

\theoremstyle{definition}
\newtheorem{definition}[thm]{Definition}

\theoremstyle{remark}
\newtheorem{rmk}[thm]{Remark}

\newcommand{\mc}{\mathcal}
\newcommand{\cA}{\mc A}

\newcommand{\cC}{\mc C}

\newcommand{\cF}{\mc F}

\newcommand{\cH}{\mc H}

\newcommand{\cK}{\mc K}
\newcommand{\cL}{\mc L}

\newcommand{\cR}{\mc R}

\newcommand{\cV}{\mc V}

\newcommand{\cX}{\mc X}

\newcommand{\ms}{\mathscr}

\newcommand{\scri}{\ms I}

\newcommand{\sR}{\ms R}

\newcommand{\C}{\mathbb{C}}
\newcommand{\N}{\mathbb{N}}
\newcommand{\R}{\mathbb{R}}

\newcommand{\Sph}{\mathbb{S}}

\newcommand{\sfG}{\mathsf{G}}

\newcommand{\bfa}{\mathbf{a}}

\newcommand{\bfq}{\mathbf{q}}

\newcommand{\fc}{\mathfrak{c}}

\newcommand{\fm}{\mathfrak{m}}

\newcommand{\ft}{\mathfrak{t}}

\newcommand{\sld}{\slashed{d}{}}

\newcommand{\slstar}{\slashed{\star}}

\newcommand{\scal}{\mathsf{S}}
\newcommand{\scalspace}{\mathbf{S}}
\newcommand{\vect}{\mathsf{V}}
\newcommand{\vectspace}{\mathbf{V}}

\newcommand{\ran}{\operatorname{ran}}
\newcommand{\ann}{\operatorname{ann}}

\renewcommand{\Re}{\operatorname{Re}}
\renewcommand{\Im}{\operatorname{Im}}

\newcommand{\mathspan}{\operatorname{span}}
\newcommand{\supp}{\operatorname{supp}}

\newcommand{\dd}{{\rm d}}

\newcommand{\tr}{\operatorname{tr}}

\newcommand{\Poly}{{\mathrm{Poly}}}

\newcommand{\CD}{{\mathrm{CD}}}
\newcommand{\cd}{\fc}

\newcommand{\Ups}{\Upsilon}

\newcommand{\eps}{\epsilon}

\newcommand{\hra}{\hookrightarrow}
\newcommand{\la}{\langle}

\newcommand{\ol}{\overline}
\newcommand{\pa}{\partial}
\newcommand{\ra}{\rangle}

\newcommand{\ul}[1]{\underline{#1}{}}

\newcommand{\wh}{\widehat}
\newcommand{\wt}{\widetilde}

\newcommand{\bop}{{\mathrm{b}}}

\newcommand{\cp}{{\mathrm{c}}}

\newcommand{\scl}{{\mathrm{sc}}}

\newcommand{\Diff}{\mathrm{Diff}}

\newcommand{\Vb}{\cV_\bop}
\newcommand{\Diffb}{\Diff_\bop}

\newcommand{\Vsc}{\cV_\scl}
\newcommand{\Diffsc}{\Diff_\scl}

\newcommand{\Tb}{{}^{\bop}T}
\newcommand{\Tsc}{{}^{\scl}T}

\newcommand{\half}{{\tfrac{1}{2}}}

\newcommand{\sub}{{\mathrm{sub}}}

\newcommand{\CI}{\cC^\infty}

\newcommand{\CIc}{\cC^\infty_\cp}

\newcommand{\Hb}{H_{\bop}}

\newcommand{\Hbext}{\bar H_{\bop}}

\newcommand{\Hbsupp}{\dot H_{\bop}}

\newcommand{\Hbh}{H_{\bop,h}}
\newcommand{\Hbhext}{\bar H_{\bop,h}}

\newcommand{\Ric}{\mathrm{Ric}}

\newcommand{\bhm}{\fm}
\newcommand{\bha}{\bfa}
\newcommand{\bhq}{\bfq}

\newcommand{\openbigpmatrix}[1]
  {%
    \def\@bigpmatrixsize{#1}%
    \addtolength{\arraycolsep}{-#1}%
    \begin{pmatrix}%
  }
\newcommand{\closebigpmatrix}
  {%
    \end{pmatrix}%
    \addtolength{\arraycolsep}{\@bigpmatrixsize}%
  }

\newlength{\enummargin}

\newcommand{\usref}[1]{{\upshape\ref{#1}}}

\DeclareGraphicsExtensions{.mps}
\newcommand{\inclfig}[1]{\includegraphics{#1}}

\makeatletter
\newcommand*{\fwbw}[1]{\expandafter\@fwbw\csname c@#1\endcsname}
\newcommand*{\@fwbw}[1]{\ifcase #1 \or {\rm fw}\or {\rm bw}\fi}
\AddEnumerateCounter{\fwbw}{\@fwbw}
\makeatother

\begin{document}

\title{Linear stability of Kerr black holes in the full subextremal range}

\date{\today}

\subjclass[2010]{}

\author{Dietrich H\"afner}
\address{Universit\'e Grenoble Alpes, Institut Fourier, 100 rue des maths, 38402 Gi\`eres, France}
\email{dietrich.hafner@univ-grenoble-alpes.fr}
\author{Peter Hintz}
\address{Department of Mathematics, ETH Z\"urich, R\"amistrasse, 8092 Z\"urich, Switzerland}
\email{peter.hintz@math.ethz.ch}
\author{Andr\'as Vasy}
\address{Department of Mathematics, Stanford University, Stanford, California 94305-2125, USA}
\email{andras@math.stanford.edu}

\date{\today}

\begin{abstract}
  We prove, unconditionally, the linear stability of the Kerr family in the full subextremal range. On an analytic level, our proof is the same as that of our earlier paper in the slowly rotating case \cite{HHV21}. The additional ingredients we use are, firstly, the mode stability result \cite{AHW22} proved by Andersson, Whiting, and the first author and, secondly, computations related to the zero energy behavior of the linearized gauge-fixed Einstein equation by the second author \cite{Hi24}.
 \end{abstract}
 \maketitle

\section{Introduction}

Stability questions for families of black hole solutions of Einstein's field equations have been intensely studied in the last decades. For small angular momenta, the nonlinear stability of the Kerr--de~Sitter family was shown by the second and third author in \cite{HiVa18} (see also \cite{FangKdS}). The nonlinear stability of the Kerr family with small angular momenta was shown in a series of papers by Klainerman--Szeftel \cite{KS23}, \cite{KS22a}, \cite{KS22b}, Giorgi--Klainerman--Szeftel \cite{GKS22} and Shen \cite{Sh23}.  We also refer to the work of Dafermos, Holzegel, Rodnianski  and Taylor \cite{DaHoRoTa22} for the codimension $3$ nonlinear stability of the Schwarzschild family. As many black holes are rapidly spinning, see e.g.\ Thorne \cite{Th74}, it is an important problem to remove the restriction to small angular momenta in the results proved thus far.

A first step in proving nonlinear stability is linear stability. In \cite{HHV21,HaefnerHintzVasyKerrErratum} we proved the linear stability of slowly rotating Kerr black holes. In the present paper we establish the same result for the full subextremal range of specific angular momenta $a$, i.e.\ $|a|<\bhm$ where $\bhm$ is the mass of the black hole. Let $(M_b^\circ, g_b)$ be the Kerr solution and $\ft$ be a time function on the manifold $M^\circ_b$ as represented in Figure \ref{FigK0Time}.  Let $\Sigma_0=\ft^{-1}(0)$. For a Lorentzian metric $g$ on $M^\circ_b$ we denote by $\tau(g)=(\gamma,k)$ the data on $\Sigma_0$, where $\gamma$ is the pullback of $g$ to $\Sigma_0$ and $k$ is the second fundamental form of $\Sigma_0$ in $(M_b^\circ,g)$. Then a rough version of our theorem reads as follows:

\begin{thm}[Linear stability of subextremal Kerr]
\label{thm3.1}
  Let $\dot{\gamma},\, \dot{k}$ be symmetric $2$-tensors on $\Sigma_0=\ft^{-1}(0)$ satisfying the {linearized constraint equations} and 
  \[
    |\dot{\gamma}| \lesssim r^{-1-\alpha}, \quad |\dot{k}| \lesssim r^{-2-\alpha}
  \]
  for some $\alpha\in(0,1)$ (and similar bounds for derivatives of $\dot\gamma,\dot k$ along $\partial_{\omega},\, r\partial_r$ up to order $8$). Let $h$ denote a solution of the linearized Einstein vacuum equation
  $D_{g_b}\Ric(h) = 0$ on $M_b^\circ$ which attains the initial data $(\dot{\gamma}, \dot{k})$ on $\Sigma_0$, i.e.\ $ D_{g_b}\tau(h)=(\dot{\gamma},\dot{k})$.\footnote{Solutions always exist, as follows from a standard reduction to a linear wave equation using gauge fixing. Solutions are however not unique, as $D_{g_b}\Ric(\cL_V g_b)=0$ for all vector fields $V$.} Then there exist linearized black hole parameters $\dot{b}=(\dot{\bhm},\dot{\bha})\in \R\times\R^3$ and a vector field $V$ on $M_b^\circ$ such that 
  \begin{align*}
  h=\dot{g}_{b}(\dot{b})+\cL_Vg_b+\tilde{h},
  \end{align*}
  where for bounded $r$ the tail $\tilde{h}$ satisfies the bound $\vert \tilde{h}\vert \le C_{\eta}\ft^{-1-\alpha+\eta}$ for all $\eta>0$. Here $\dot g_b(\dot b):=\frac{\rm d}{{\rm d}s}g_{b+s\dot b}|_{s=0}$ is a linearized Kerr metric.
\end{thm}
    
  In a fixed generalized harmonic gauge and replacing $\dot{g}_b(\dot{b})$ by its gauge-fixed version we get $\ft^{-\alpha+\eta}$ decay and $V$ can be chosen in a fixed $6$-dimensional space. If we permit in addition a gauge source term, we can obtain the stated $\ft^{-1-\alpha+\eta}$ decay. We refer to \S\ref{Sec8} for precise versions of these results.

In \cite{HHV21} the main object of study was the forward forcing problem for the gauge-fixed linearized Einstein operator around a given Kerr metric $g_b$, which we write as 
\begin{align}
\label{LEI}
L_bh=f.
\end{align}
We show that for a well-chosen gauge, the solution can be expressed as $h=g^{\prime\Ups}_b(\dot{b})+\cL_Vg_b+\tilde{h}$, where $g_b^{\prime\Ups}(\dot{b})$ is a linearized Kerr solution (in the chosen gauge), the vector field $V$ lies in a fixed $6$ dimensional space consisting of asymptotic translations and asymptotic boosts and $\tilde{h}$ decays in time. In \cite{HHV21} we have chosen a generalized harmonic gauge. To obtain our result we Fourier transform \eqref{LEI} and study the Fourier transformed gauge-fixed linearized Einstein operator $\wh L_b(\sigma)$ (obtained from $L_b$ by replacing time derivatives by multiplication by $-i\sigma$).

The main ingredients of our proof of the result in \cite{HHV21} were the following:
\begin{enumerate}
\item A \textit{robust Fredholm framework} for the operators $\wh{L_b}(\sigma)$. In \cite{HHV21}, we construct suitable function spaces such that the operator $\wh{L_b}(\sigma)$, $\Im\sigma\geq 0$, acts as a Fredholm operator of index $0$ between them. 
\item \textit{High energy estimates}, i.e.\ estimates for the resolvent $\wh{L_b}(\sigma)^{-1}$ for large $\vert\Re\sigma\vert$ and bounded $\Im\sigma$. These estimates only use the structure of the trapping which is $r$-normally hyperbolic for every $r$, as shown in \cite{WZ11} (small $|a|/\bhm$), \cite{Dy15} (full subextremal range). 
\item \textit{Uniform Fredholm estimates down to $\sigma=0$} and \textit{regularity of the resolvent} at $\sigma=0$. These estimates mainly use the asymptotic behavior of the metric at infinity. 
We rely on the estimates proved in \cite{Va21a,Va21b}. 
\item \textit{Mode stability} for $L_b$, i.e.\ invertibility of $\wh{L_b}(\sigma)$ for $\sigma\neq 0,\, \Im\, \sigma\ge 0$ and a precise understanding of the $0$ modes (elements of the kernel of $\wh{L_b}(0)$).
\end{enumerate}  
Points (1)--(3) depend solely on structural properties of the Kerr metric: they use only the asymptotic behavior of the metric or the dynamical properties of the null-geodesic flow (in particular: horizons, trapping). The important fact is that these features remain unchanged in the large $a$ case. This is obvious from the form of the metric for the asymptotic behavior at infinity, and has been checked by Dyatlov \cite{Dy15} concerning the $r$-normally hyperbolic trapping.

Point (4) above was, to a large extent, carried out by Andersson, Whiting, and the first author in \cite{AHW22} in the large $a$ case. However, we need more here: \cite{AHW22} does not concern generalized modes (i.e.\ solutions of $L_b h=0$ which are polynomials in time); and \cite{AHW22} uses a standard wave map gauge with standard implementation, which is not suitable for our present needs. The first reason is that the standard implementation allows quadratically growing (in time) modes.\footnote{Note that boosts are linearly growing, so linearly growing modes must be permitted.} These modes have been eliminated in \cite{HHV21} by means of constraint damping; we use a similar procedure in this paper. Secondly, linearized Kerr solutions with respect to the mass cannot be put into the gauge while retaining stationarity; they instead showed up in \cite{HHV21} as linearly growing generalized mode solutions. We therefore follow the strategy of the second author in \cite{Hi24}, where he performs both a gauge modification and constraint damping. This does not only solve the aforementioned problems, but also eliminates the stationary Coulomb type mode solutions which still appeared in \cite{HHV21}.

We mention here only a few previous results concerning the linearized Einstein and the Teukolsky equations; we refer the reader to \cite{HHV21} for a more detailed literature review. Decay estimates for the Teukolsky equation on slowly rotating Kerr black holes were shown by Dafermos--Holzegel--Rodnianski \cite{DaHoRo19a}. In the same context, Ma--Zhang \cite{MaZh23} obtain optimal decay estimates with computation of the first term. Millet \cite{Mi23} obtains equivalent results to those of Ma--Zhang in the large $a$ case. Shlapentokh-Rothman and Teixeira da Costa \cite{SRTdC20,SRTdC23} obtain energy boundedness and integrated local energy decay results in this context. 

Linear stability was shown for the Schwarzschild metric by Dafermos, Holzegel, and Rodnianski \cite{DaHoRo19} and Hung--Keller--Wang \cite{HKW20}. Linear stability was shown for slowly rotating Kerr black holes by Andersson--B\"ackdahl--Blue--Ma  \cite{ABBM19} and the authors \cite{HHV21}. The result \cite{ABBM19} in fact gives a \emph{conditional} linear stability proof for the full subextremal range which \emph{assumes} certain integrated local energy decay estimates for the Teukolsky equations. (These estimates were proved for small $a$ by Ma \cite{MaGravityKerr}.)
 
In contrast to the Kerr case, the mode analysis has not been carried out in the Kerr--de~Sitter case in the full subextremal range. The mode analysis \cite{AHW22} in the Kerr case makes decisive use of the mode stability result of Whiting \cite{Wh89} for the Teukolsky equation and its sharpenings \cite{ShlapentokhRothmanModeStability,TeixeiradCModes,AMPW17}. An equivalent result is not known even for the wave equation on Kerr--de~Sitter: in this case mode stability was shown using perturbative arguments for small angular momenta \cite{Dy11a} and for small masses \cite{HintzKdSMS}. Conditional on mode stability for the (ungauged) linearized Einstein equation, the second and third author with Petersen prove the nonlinear stability of Kerr--de~Sitter in the full subextremal range in~\cite{HintzPetersenVasyKdS}.
 
 The paper is organized as follows. In~\S\ref{Sec2} we give a description of the Kerr metric. Gauge issues and the strategy of the proof are presented in~\S\ref{Sec3}. We revisit b- and scattering structures in~\S\ref{Sec4}. Our Fredholm setting is presented in~\S\ref{Sec6}. The mode analysis is performed in~\S\ref{Sec7}. \S\ref{Sec9} is devoted to the analysis of the regularity of the resolvent. Precise versions of our main result are given in~\S\ref{Sec8}.  
 
 \noindent
{\bf Acknowledgements.} D.H.~gratefully acknowledges support from the ANR funding  ANR-20-CE40-0018-01. P.H.~is grateful to the Swiss National Science Foundation under grant number TMCG-2\textunderscore{}223167 / 1. A.V.~gratefully acknowledges support from the National Science Foundation under grant number DMS-2247004.

\section{The Kerr family}
\label{Sec2}

We parameterize the Kerr family of solutions to the Einstein vacuum equations by $b=(\bhm,\bha)\in\R^+\times\R^3$, with $\bhm$ denoting the mass and $\bha$ the specific angular momentum. We restrict to the subextremal case $a=\vert \bha\vert< \bhm$. Let 
\begin{equation*}
r_\pm:=\bhm \pm\sqrt{\bhm^2-a^2}. 
\end{equation*}
Let $(\theta, \phi) \in [0,\pi) \times [0,2\pi)$ be coordinates on $\Sph^2$ such that $\theta=0$ points in the direction of $\bha$ when $\bha\neq 0$. The Kerr metric $g_b^{\rm BL}$ in Boyer-Lindquist coordinates $(t,r,\theta,\varphi) \in \mathbb{R} \times (r_+ , \infty) \times \Sph^2$ is then
\begin{subequations}
\label{EqKaMetric}
\begin{align}
  g_b^{\rm BL} &= \frac{\Delta_b}{\varrho_b^2}(\dd t-a\sin^2\theta\,\dd\varphi)^2 - \varrho_b^2\Bigl(\frac{\dd r^2}{\Delta_b}+\dd\theta^2\Bigr) - \frac{\sin^2\theta}{\varrho_b^2}\bigl(a\,\dd t-(r^2+a^2)\dd\varphi\bigr)^2, \\
  \intertext{with inverse} 
G_b^{\rm BL} &= \frac{1}{\Delta_b\varrho_b^2}\bigl((r^2+a^2)\pa_t+a\pa_\varphi\bigr)^2 - \frac{\Delta_b}{\varrho_b^2}\pa_r^2 - \frac{1}{\varrho_b^2}\pa_\theta^2 - \frac{1}{\varrho_b^2\sin^2\theta}(\pa_\varphi+\bha\sin^2\theta\,\pa_t)^2, \\
  \quad \Delta_{b} &= r^2-2\bhm r+a^2, \ \ 
  \varrho_{b}^2 = r^2+a^2\cos^2\theta. 
  \end{align}
\end{subequations}
Note that our signature convention is $(+,-,-,-)$. $g_b^{\rm BL}$ is the metric of an isolated, rotating, stationary black hole with event horizon $\{ r=r_+ \}$, and the domain of outer communications is $\{r > r_+\}$. Further, $\{r=r_-\}$ is the inner horizon.\footnote{Strictly speaking, the event horizon and the inner horizon are $r=r_{\pm}$ in another coordinate system, see below.} We have  $\sqrt{|\det g_b|} = \varrho_b^2 \sin\theta$. The Kerr metric is a solution of the Einstein vacuum equation:
\begin{equation}
\label{EqKaEinstein}
  \Ric\bigl(g_{b}^{\rm BL}\bigr)=0.
\end{equation}
The form~\eqref{EqKaMetric} of the metric breaks down at $r=r_{\pm}$ which are the roots of $\Delta_b$. Let 
\begin{equation*}
r_*=\int \frac{r^2+a^2}{\Delta_b} \dd r\quad \mbox{with}\quad r_*(3\bhm)=0. 
\end{equation*}
Consider coordinates $(t_*,r,\varphi_*,\theta)$, where $t_*(t,r)=t+F(r)$ and $\varphi_*(\varphi,r)=\varphi+T(r)$ are smooth functions such that 
\begin{align*}
F(r)&=\begin{cases}%
r_*&\mbox{for } r\le 3\bhm,\\
-r_*&\mbox{for } r\ge 4\bhm, 
\end{cases}\\
T(r)&=\begin{cases}%
\int \frac{a}{\Delta_b}&\mbox{for } r\le 3\bhm,\, T(3\bhm)=0,\\
0&\mbox{for } r\ge 4\bhm. 
\end{cases}%
  \end{align*}
For $r\le 3\bhm$, the metric then takes the form
\begin{equation}
\label{EqKerrMetric}
\begin{split}
g_b&=\frac{\Delta_b}{\varrho_b^2}(\dd t_*-a\sin^2\theta \dd\varphi_*)^2-2(\dd t_*-a\sin^2\theta \dd\varphi_*)\dd r-\varrho_b^2\dd\theta^2\\
&\quad -\frac{\sin^2\theta}{\varrho_b^2}(a \dd t_*-(r^2+a^2)\dd\varphi_*)^2,
\end{split}
\end{equation}
which is clearly smooth up to $r_+$.

Given
\begin{equation}
\label{eqr0}
r_- < r_0 < r_+ ,
\end{equation}
let
\[
X^\circ_b = [r_0,\infty)\times \Sph^2 ,
\]
and consider 
\begin{align}
\label{MX}
M^\circ_b = \R \times X^\circ_b, 
\end{align}
with coordinates $t_*, r, \theta, \varphi_*$. We choose the dependence of the functions $T,F$ above on $b$, near a fixed value of $b_0=(\bhm,\bha)$, such that $g_b$ is a smooth family of stationary metrics on $M^\circ_{b_0}$.

We compactify $X^\circ_b$ as follows: we set
\begin{equation*}
X_b:=\ol{X^\circ_b}\subset\ol{\R^3},\quad \rho:=\frac{1}{r},
\end{equation*}
where $\ol{\R^3}$ denotes the radial compactification (see~\S\ref{Sec4} for details). Thus, $X_b=\ol{\{r\geq r_0\}}$, and we let $\pa_-X_b=r^{-1}(r_0)$, $\pa_+X_b=\pa\ol{\R^3}\subset X$. Within $X_b$, the topological boundary of the domain of outer communications $\cX=(r_+,\infty)\times\Sph$ has two components,
\begin{equation}
\label{EqK0StaticBdy}
  \pa\cX = \pa_-\cX \sqcup \pa_+\cX,\qquad
  \pa_-\cX := r^{-1}(r_+), \quad
  \pa_+\cX := \rho^{-1}(0) = \pa_+ X.
\end{equation}
Note that $\pa_-X_b$ is distinct from $\pa_-\cX$, and is indeed a hypersurface lying \emph{beyond} the event horizon. Note that $\R \times\pa X$ has two components,
\begin{equation}
\label{EqK0SurfFin}
  \Sigma_{\rm fin} := \R \times\pa_- X_b
\end{equation}
(which is a spacelike hypersurface inside of the black hole) and $\R_{t_*}\times\pa_+ X_b=\R_{t_*}\times\pa_+\cX$ (which is future null infinity, typically denoted $\scri^+$); moreover, the future event horizon, $\cH^+$, is $\R_{t_*}\times\pa_-\cX$. 

We will also need a function $\ft\in\CI(M^\circ_b)$ of the form $\ft=t_*+G(r),$ which satisfies $\ft=t$ for $r\geq 4\bhm$, and such that $\dd\ft$ is future timelike on $M_b^\circ$ with respect to $g_{b}$.

\begin{figure}[!ht]
  \centering
    \inclfig{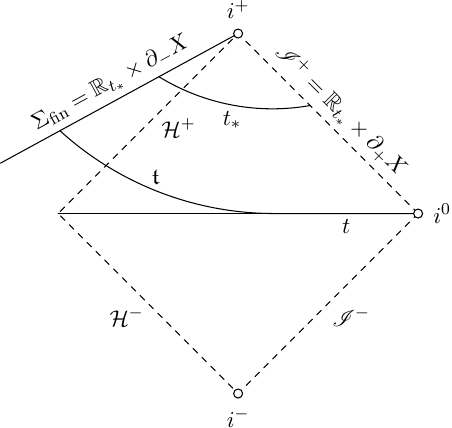}
  \caption{Illustration of time functions on $M^\circ_b$ in the Penrose diagram of the Kerr metric (including future/past null infinity $\scri^\pm$, the future/past event horizon $\cH^\pm$, spacelike infinity $i^0$, and future/past timelike infinity $i^\pm$). Shown are level sets of  the functions $\ft$ and $t_*$. We also indicate the boundaries $\R_{t_*}\times\pa_+ X$ (future null infinity again) and $\R_{t_*}\times\pa_- X$ (a spacelike hypersurface beyond the future event horizon). }  
\label{FigK0Time}
\end{figure}

Finally, we recall the that linearized Kerr metrics
\begin{equation}\label{gdoteqn}
\dot{g}_b(\dot{b})=\left.\frac{d}{ds}\right\vert_{s=0}g_{b+s\dot{b}},\quad \dot{b}\in \R^4,
\end{equation}
are solutions of the linearized Einstein equations 
\begin{align*}
D_{g_b}\Ric(\dot{g}_b(\dot{b}))=0.
\end{align*}

\section{Gauge issues and strategy of the proof}
\label{Sec3}

\subsection{Gauge issues}

To eliminate the diffeomorphism the invariance (if $\Ric(g)=0$, then also $\Ric(\phi^*g)=0$ for all diffeomorphisms $\phi$), we impose an extra condition on $g$. We introduce a gauge $1$-form $\Upsilon(g;g_b)$ which vanishes if and only if $\mathbf{1}: (M_b^\circ,g)\rightarrow (M_b^\circ,g_b)$ is a wave map: 
\[
  \Upsilon(g;g_b) = 0\iff \Box_{g,g_b}\mathbf{1} = 0.
\]
Explicitly, $\Upsilon(g;g_b)$ is given by
\[\Upsilon(g;g_b)=gg_b^{-1}\delta_g\sfG_gg_b.\]
Here $(\delta_g h)_\mu=-h_{\mu\nu}{}^{;\nu}$ is the adjoint of the symmetric gradient $(\delta_g^*\omega)_{\mu\nu}=\frac12(\omega_{\mu;\nu}+\omega_{\nu;\mu})$, and $\sfG_g=\mathbf{1}-\frac{1}{2}g\tr_g$ is the trace reversal operator in $4$ spacetime dimensions. The gauge $\Ups(g;g_b)=0$ is often called a {\it wave map gauge} (or generalized harmonic gauge). The {\it DeTurck trick} then consists in the realization of an equivalence between
\begin{align*}
  &\begin{cases}
    \Ric(g)=0, \\
    \Upsilon(g;g_b) = 0, \\
    \text{initial data satisfying the constraint equations}
  \end{cases}
  \end{align*}
  and an initial value problem (with suitably constructed initial data) for the equation 
  \begin{align}
  \label{P(g)}
  P(g):=\Ric(g)-\delta_g^*\Upsilon(g;g_b) = 0,
\end{align}
such that the initial data  $\Upsilon(g;g_b)$ are zero. The first equation implies the second one. To go from the second to the first one, one applies the second Bianchi identity: this shows that $\Upsilon(g;g_b)$ solves a homogeneous wave equation with zero initial data (the vanishing of the normal derivatives of $\Ups(g;g_b)$ at the initial surface uses the constraint equations) and therefore has to be zero. The advantage of~\eqref{P(g)} over $\Ric(g)=0$ is that the former is a quasilinear wave equation, whereas the latter is not.
 
The linearized version of this is as follows. The system
\begin{align*}
  &\begin{cases}
    D_{g_b}\Ric(h)=0, \\
    D_{g_b}\Upsilon(g;g_b)(h) = 0, \\
    \text{initial data satisfying the linearized constraint equations}
  \end{cases}
  \end{align*}
  is equivalent to an initial value problem for 
  \begin{align*}
    D_{g_b}P(h)=(D_{g_b}\Ric+\delta^*_{g_b}\delta_{g_b}\sfG_{g_b})(h)=0.
  \end{align*}
  We note that $D_{g_b}P(h)=\half\Box_{g_b}h$ plus first order terms: this is thus a wave equation on $(M_b^\circ,g_b)$.

For the purpose of analyzing solutions of $D_{g_b}P(h)=0$ globally and obtaining good decay for $h$, it turns out that one has to modify this operator. To motivate this, let us consider the kinds of solutions we expect. The operator $D_{g_b}P$ has natural mode solutions corresponding to linearized Kerr, asymptotic translations as well as asymptotic boosts at infinity (the asymptotic rotations are contained in the linearized Kerr solutions, see \cite[Remark~9.2]{HHV21} for details). However for the above operator additional mode solutions exist which are not natural. 
\begin{enumerate}
\item There exists an additional mode solution which comes from a Coulomb type solution of the $1$-form wave operator. Moreover, linearizations $\dot g_b(1,0)$ of the Kerr metric with respect to changes only in the mass cannot be put into the gauge $\delta_{g_b}\sfG_{g_b}(\dot g_b(1,0)+\delta_{g_b}^*\omega)=0$ while remaining stationary; see \cite[Proposition~9.4]{HHV21}. 
\item There exist generalized mode solutions which are quadratically growing in time. These are however pathological in the sense that they do not satisfy the (not gauge-fixed) linearized Einstein equations; see \cite[\S{9.3}]{HHV21}.
\end{enumerate}
To overcome the second difficulty, we follow \cite{HHV21}. The first difficulty does not need to be addressed (and indeed was not addressed in \cite{HHV21}) to achieve a stability proof, but we address it here nonetheless, following \cite{Hi24}, as this results in a simpler version of the main linear stability result.

First observe that there is some flexibility in the above DeTurck approach.  
\begin{enumerate}
\item We modify the divergence operator $\delta_{g_b}$ by a zeroth order term,
\begin{align*}
\delta_{g_{b},\gamma_{\Upsilon}}=\delta_{g_b}+\gamma_{\Upsilon}E_{\Upsilon}. 
\end{align*}
By doing so, we change the gauge. We will choose 
\begin{align*}
E_{\Upsilon}=2{i}_{{\cd_{\Upsilon}^{\#}}}h-\cd_{\Upsilon}\tr_{g_b}h
\end{align*}
for a suitable $\cd_{\Upsilon}\in C_c(X_b^\circ;T^*_{X_b^\circ}M^\circ_b)$. Here the vector field $\cd_{\Upsilon}^{\#}$ is defined by $\cd_{\Upsilon}^{\#}=g^{-1}(\cd_{\Upsilon},\cdot)$. 
\item We replace the symmetric gradient $\delta_{g_b}^*$ by 
\begin{align*}
\delta_{g_b,\gamma_C}^*:=\delta_{g_b}^*+\gamma_C E_\CD.
\end{align*}
By doing so, we change the \emph{implementation} of the gauge; such a procedure is called constraint damping \cite{BrodbeckFrittelliHubnerReulaSCP,GundlachCalabreseHinderMartinConstraintDamping}. We choose 
\begin{align}
E_\CD\omega=2\cd_C\otimes_s\omega-\la \cd_C,\omega\ra_{g^{-1}_b}g_b 
\end{align}
for a suitable $\cd_{C}\in C_c(X_b^\circ;T^*_{X_b^\circ}M^\circ_b)$. 
\end{enumerate}

The quantities and $1$-forms $\gamma_{\Upsilon}$, $\gamma_C$, $\cd_{\Upsilon}$, and $\cd_C$ will be fixed later on. For fixed choices, we then write
\begin{align*}
L_b:=2(D_{g_b}\Ric+\delta^*_{g_b,\gamma_{C}}\delta_{g_b,\gamma_{\Upsilon}}\sfG_{g_b}).
\end{align*}
\begin{rmk}
(Thus, $L_b=\Box_{g_b}$ plus lower order terms, acting on symmetric $2$-tensors.) An important point is that $\cd_{C}$ and $\cd_{\Upsilon}$ are compactly supported. Whereas the different operators are changed with respect to \cite{HHV21} and \cite{AHW22}, the normal (model) operators at $r=\infty$ are not changed. A large part of our analysis is based on normal operator arguments and this part is completely unchanged. 
Threshold regularities, which are the minimal regularity required for the underlying estimates at the conormal bundle of the event horizon in the `spatial' manifold $X_b^\circ$ (see e.g.\ \cite[Chapter~11]{HintzMicro}), can in principle change by these modifications. For the unmodified gauge-fixed operator we check in the appendix that the threshold regularities are still the same as in the slowly rotating case.  Once $\cd_{\Upsilon}$ and $\cd_C$ are fixed, the threshold regularity depends continuously on $\gamma_{\Upsilon}, \, \gamma_C$ and these constants can be chosen arbitrary small in the present paper. As a result in our setting, we can still work with the same threshold regularities as in the small $a$ case. 
\end{rmk}

We will now study $L_b h=0$ with general initial data, i.e.\ it is not necessary that the initial data satisfy the linearized constraint equations. (In previous work on Kerr--de~Sitter \cite{HiVa18}, this additional flexibility was important for the proof of \emph{nonlinear} stability.)

\begin{thm}[Asymptotics for the linearized gauge-fixed Einstein equation]
\label{thm3.2}
  Let $\alpha\in(0,1)$, and let $h_0,h_1\in\CI(\Sigma;S^2 T_\Sigma^*M)$,
  \[
    |h_0| \lesssim r^{-1-\alpha},\quad
    |h_1| \lesssim r^{-2-\alpha}
  \]
  for $\alpha\in(0,1)$ (and similar bounds for derivatives). Let
  \[
    \begin{cases}
      L_{b} h = 0, \\
      (h|_\Sigma,\ \cL_{\pa_t}h|_\Sigma) = (h_0,h_1).
    \end{cases}
  \]
  Then there exist {$\dot{b}\in\R\times\R^3$} and vector fields $V_1$, $V_2$ such that
    \[
    h = g_b^{\prime\Ups}(\dot{b}) + \cL_{V_1}g_b + \cL_{V_2}g_b + \tilde h,\quad
    {|\tilde h|\lesssim t_*^{-1-\alpha+},\quad \vert \cL_{V_2} g_b\vert \lesssim t_*^{-\alpha+}.}
  \]
  Here $g_b^{\prime\Ups}(\dot{b})$ is a gauge-fixed version of $\dot{g}_b(\dot{b})$ (i.e.\ $\delta_{g_b,\gamma_\Ups}\sfG_{g_b}g_b^{\prime\Ups}(\dot b)=0$) and $V_1$ can be chosen to lie in a fixed 6-dimensional space consisting of symmetric gradients of suitable asymptotic translations and asymptotic boosts.  
 \end{thm}
 Using this theorem one can therefore read off the change of black hole parameters and the movement of the black hole in the chosen gauge. As already mentioned in the introduction, one can eliminate the vector field $V_2$ if one changes the gauge further using an appropriate gauge source term (constructible from $V_2$, namely given by $\delta_{g_b,\gamma_\Ups}\sfG_{g_b}\cL_{V_2}g_b$).  We refer to Theorems \ref{thm7.2}, \ref{thm7.3} for a more precise version of this theorem. 

We recall once more the work of Andersson--B\"ackdahl--Blue--Ma  \cite{ABBM19} which is unconditional for slowly rotating Kerr black holes: this establishes linear stability for initial data with fast decay ($\alpha>\frac52$). This eliminates in particular the emergence of linearized Kerr solutions (i.e.\ $\dot{b}=0$, see also \cite[Remark~14.7]{HHV21}).

\subsection{Strategy of the proof}

\label{Sec3.3}
We follow the same strategy as in \cite{HHV21}. We first recast the problem as a forcing problem. The initial value problem can be reduced to a forcing problem 
\begin{align}
\label{4.1}
L_b h = f,
\end{align}
where $\supp f\subset \{t_*>0\}$ and we are looking for a forward solution in the sense that $\supp u\subset\{t_*>-T\}$ for some $T>0$. (See also \cite{Mi23}.) Fourier transforming \eqref{4.1} with respect to $t_*$ gives 
\begin{align*}
\wh{L_b}(\sigma)\hat{h}(\sigma)=\hat{f}(\sigma). 
\end{align*}
We then work on the spectral side:
\[
  h(t_*,x) = \frac{1}{2\pi}\int_{\Im\sigma=C} e^{-i\sigma t_*}{\wh{L_b}(\sigma)^{-1}}\hat f(\sigma,x)\,\dd\sigma.
\]
This formula is valid for $C\gg 1$ (amounting, roughly speaking, to an a priori bound of the forward solution by $e^{C t_*}$ which follows from a simple energy estimate). The idea is now to shift contour to $C=0$. The main ingredients of the proof are then:
\begin{enumerate}
\item Fredholm theory for non-elliptic operators, ultimately based on \cite{Va13}. More precisely, the operator $\wh{L_b}(\sigma)$ is Fredholm (of index $0$) as an operator between suitable function spaces based on weighted Sobolev spaces. Other than \cite{Va13}, this uses scattering (radial point) estimates at infinity \cite{Me95, Va21a} (for non-zero $\sigma$), radial point estimates at the horizons \cite{Va13}, real principal type propagation of regularity \cite{DuHo}, and (for $\sigma=0$) elliptic b-theory \cite{Me94}. We will make our functional setting more precise in~\S\ref{Sec6}.
\item Mode analysis of the gauge-fixed Einstein equations and certain $1$-form wave operators. 
\item Constraint damping as already discussed. 
\item $\wh{L_b}(\sigma)$ satisfies high energy estimates (in particular: is invertible) for $|\Re\sigma|\gg 1$ and bounded $\Im\sigma\geq 0$. This uses semiclassical estimates at the horizons \cite{Va13} and spatial infinity \cite{VZ00,Va21a}, as well as estimates at normally hyperbolic trapping \cite{WZ11,Dy15,Dy16,HintzPsdoInner} (see also \cite{NZ15}).
\item Low frequency resolvent estimates, meaning uniform Fredholm estimates for $\wh{L_b}(\sigma)$ down to $\sigma=0$ \cite{Va21b}; see also \cite{BoHa10,VW13,HintzNonstat}. 
\end{enumerate}

\section{b- and scattering structures}
\label{Sec4}

We recall geometric structures and variants of tangent and cotangent bundles on manifolds with boundary, arising in the present paper in the description of the asymptotic behavior of operators, $1$-forms, and symmetric $2$-tensors near `infinity' in compactifications of the spatial or spacetime manifold.

Let $X$ denote an $n$-dimensional smooth manifold with boundary $\pa X\neq\emptyset$. Let $\rho\in\CI(X)$ be a boundary defining function, which means that $\rho\geq 0$ everywhere, $\pa X=\rho^{-1}(0)$, and $\dd\rho\neq 0$ on $\pa X$. We then define
\[
  \Vb(X) := \{ V\in\cV(X)=\CI(X;T X) \colon X\ \text{is tangent to}\ \pa X \},\quad
  \Vsc(X) := \{ \rho V \colon V\in\Vb(X) \}.
\]
In local coordinates $x\geq 0$, $y\in\R^{n-1}$ near a point in $\pa X$, with $\pa X$ locally given by $x=0$, b-vector fields $V\in\Vb(X)$, resp.\ scattering vector fields $V\in\Vsc(X)$, are precisely those smooth vector fields which are linear combinations, with $\CI(X)$ coefficients, of $x\pa_x$, $\pa_{y^j}$ ($j=1,\ldots,n-1$), resp.\ $x^2\pa_x$, $x\pa_{y^j}$. We define vector bundles
\[
  \Tb X \to X,\quad \Tsc X \to X
\]
to have these vector fields as local frames (down to $\pa X$), while for $p\in X^\circ$ we set $\Tb_p X=\Tsc_p X=T_p X$. The dual bundles are denoted $\Tb^*X,\Tsc^*X\to X$, with local frames $\frac{\dd x}{x}$, $\dd y^j$ and $\frac{\dd x}{x^2}$, $\frac{\dd y^j}{x}$, respectively. The space $\Diffb^m(X)$ consists of all $m$-th order differential operators which are locally finite sums of up to $m$-fold compositions of elements of $\Vb(X)$ (with $0$-fold compositions being multiplication operators by elements of $\CI(X)$); analogously for $\Diffsc^m(X)$. For a weight $a\in\R$, we moreover define the class $\rho^a\Diffb^m(X)=\{\rho^a L\colon L\in\Diffb^m(X)\}$ of weighted b-differential operators.

An important example of this setup arises for $X=\ol{\R^n}$ where
\[
  \ol{\R^n} := \bigl(\R^n \sqcup ([0,1)_\rho\times\Sph^{n-1})\bigr) / \sim,\quad 0\neq x=r\omega \sim (r^{-1},\omega),
\]
is the radial compactification of $\R^n$. A change of variables computation shows that $\Vsc(\ol{\R^n})$ is the space of linear combinations of the coordinate vector fields $\pa_{x^j}$ with $\CI(\ol{\R^n})$ coefficients. Elements of $\CI(\ol{\R^n};S^2\,\Tsc^*\ol{\R^n})$ are thus of the form $g_{i j}(x)\,\dd x^i\otimes_s\dd x^j$ (where $\omega\otimes_s\eta=\frac12(\omega\otimes\eta+\eta\otimes\omega)$) for $g_{i j}\in\CI(\ol{\R^n})$; the Euclidean metric is an example. The space of b-vector fields on $\ol{\R^n}$ is spanned over $\CI(\ol{\R^n})$ by $\pa_{x^j}$ and $x^i\pa_{x^j}$, $1\leq i,j\leq n$. (Outside any neighborhood of the origin, $\pa_{x^j}$ can be dropped.)

In the present paper, we will encounter `spatial' manifolds $X$ arising as hypersurfaces inside of a `spacetime' manifold $\R\times X$. Denoting the first coordinate by $t$, we then define the bundle
\begin{equation}
\label{Eq4wtTsc}
  \wt\Tsc{}^*X \to X
\end{equation}
by $\wt\Tsc{}^*_x X:=\R\oplus\Tsc^*_x X$; a section $v=(v_0,w)\in\CI(X;\wt\Tsc{}^*X)$ is then identified with the stationary $1$-form $v_0\,\dd t+w$ on $\R\times X$.

We next discuss b- and scattering Sobolev spaces on manifolds $X$ with boundary; we now assume $X$ to be compact. We fix a \emph{scattering} density, i.e.\ a positive section of $|\Lambda^n\,\Tsc^*X|$; in local coordinates as above, this is of the form $a(x,y)|\frac{\dd x}{x^2}\frac{\dd y^1}{x}\cdots\frac{\dd y^{n-1}}{x}|$ for $0<a\in\CI$. (On $X=\ol{\R^n}$, an example is the standard Euclidean density.) For $k\in\N_0$, we now define
\[
  \Hb^k(X) := \{ u\in L^2(X) \colon A u\in L^2(X)\ \forall\,A\in\Diffb^k(X) \},\quad
  \Hb^{k,\ell}(X) := \{ \rho^\ell u \colon u\in\Hb^k(X) \}.
\]
Using a finite spanning set $\{A_\alpha\}$ of $\Diffb^k(X)$ and a fixed choice of boundary defining function $\rho$, we can give $\Hb^{k,\ell}(X)$ the structure of a Hilbert space with squared norm $\sum_\alpha \|\rho^{-\ell}A_\alpha u\|_{L^2(X)}^2$. The space $\Hb^{s,\ell}(X)$ can be defined for real $s\in\R$ using duality and interpolation. We moreover set
\[
  \Hb^{\infty,\ell}(X) := \bigcap_{s\in\R} \Hb^{s,\ell}(X),\quad
  \Hb^{-\infty,\ell}(X) := \bigcup_{s\in\R} \Hb^{s,\ell}(X).
\]
For $s\in\R\cup\{\pm\infty\}$, we moreover set
\begin{equation}
\label{EqHbpm}
  \Hb^{s,\ell+}(X) := \bigcap_{\eps>0} \Hb^{s,\ell+\eps}(X),\quad
  \Hb^{s,\ell-}(X) := \bigcup_{\eps>0} \Hb^{s,\ell-\eps}(X).
\end{equation}
Similarly to the $L^2$-based spaces $\Hb^{\infty,\ell}(X)$, we shall frequently encounter the $L^\infty$-based spaces
\[
  \cA^\ell(X) := \{ u\in\CI(X^\circ) \colon \rho^{-\ell}A u\in L^\infty(X^\circ)\ \forall\,A\in\Diffb^k(X),\ k\in\N_0 \}
\]
of \emph{conormal functions}; spaces $\cA^{\ell+}(X)$ and $\cA^{\ell-}(X)$ are defined analogously to~\eqref{EqHbpm}. Sobolev embedding for weighted b-Sobolev spaces on $X=\ol{\R^3}$ reads
\[
  \Hb^{s,\ell}(X) \subset \la x\ra^{-\ell-\frac32}L^\infty(X),\ s>\tfrac32;\quad
  \Hb^{\infty,\ell}(X) \subset \cA^{\ell+\frac32}(X).
\]
The shift of $\frac32$ arises from our use of a scattering density to define $L^2$. Spaces of sections of vector bundles are defined using local trivializations as usual.

For high energy estimates, we will work with semiclassical b-Sobolev spaces: for $h>0$, we put $\Hbh^{s,\ell}(X)=\Hb^{s,\ell}(X)$ as a set, but the norm depends on the semiclassical parameter $h\in(0,1]$: if $V_1,\ldots,V_N\in\Vb(X)$ span $\Vb(X)$ over $\CI(X)$, we let
\[
  \|u\|_{\Hbh^k(X)}^2 := \sum_{\genfrac{}{}{0pt}{}{1\leq i_1,\ldots,i_j\leq N}{0\leq j\leq k}} \|(h V_{i_1})\cdots(h V_{i_j})u\|_{L^2(X)}^2,\quad
  \|u\|_{\Hbh^{k,\ell}(X)} := \|\rho^{-\ell}u\|_{\Hbh^k(X)},
\]
for $k\in\N_0$. For real orders $k$, we use duality and interpolation to define the norms.

In our application, we shall encounter the following situation: we are given a smooth manifold $X'$ with boundary (later taken to be $X'=\ol{\R^3}$), and a smooth top dimensional submanifold $X\subset X'$ whose boundary is a disjoint union
\[
  \pa X = \pa_-X \sqcup \pa_+X,\quad \pa_+X=\pa X';
\]
thus $\pa_+X$ is the boundary `at infinity' and $\pa_-X$ is an interior boundary (later chosen to be a level set of $r=|x|$). By an abuse of notation, we then set $\Vb(X):=\{V|_X\colon V\in\Vb(X')\}$, that is, we require b-behavior of vector fields at $\pa_+X$, but at $\pa_-X$ only require smoothness; similarly for $\Vsc(X)$. In the presence of the interior boundary $\pa_-X$, there are two classes of Sobolev spaces (we only discuss the b-case here, the scattering case being completely analogous): spaces of extendible distributions,
\[
  \bar H_\bop^{s,\ell}(X) := \{ u|_{X^\circ} \colon u\in\Hb^{s,\ell}(X') \},
\]
and spaces of supported distributions
\[
  \dot H_\bop^{s,\ell}(X) := \{ u\in\Hb^{s,\ell}(X') \colon \supp u\subset X \}.
\]
(The terminology comes from \cite[Appendix~B]{Hoe85}.) These spaces are equipped with the quotient, resp.\ subspace topology, and are thus again Hilbert spaces. Semiclassical analogues of these spaces are defined completely analogously.

Finally, if $X$ is again the `spatial' manifold of a `spacetime' manifold $M=\R_\ft\times X$, and if $\cF(X)$ denotes a space of distributions on $X^\circ$, then we let
\[
  \Poly^k(\ft)\cF(X) := \Biggr\{ \sum_{j=0}^k \ft^j a_j \colon a_j\in\cF(X) \Biggr\},\quad
  \Poly(\ft)\cF(X) := \bigcup_{k\in\N_0} \Poly^k(\ft)\cF(X).
\]


\section{Fredholm setting and resolvent estimates}
\label{Sec6}

In addition to the linearized Einstein operator there appear, in the analysis of $L_b$, two further wave operators acting on $1$-forms:
\begin{enumerate}
\item The {\it gauge potential wave operator}
\begin{align*}
\Box^{\Upsilon}_{g_b,\gamma_{\Upsilon}}:=2\delta_{g_b,\gamma_{\Upsilon}}\sfG_{g_b}\delta^*_{g_b}.
\end{align*}
(If $\Box_{g_b,\gamma_\Ups}^\Ups\omega=0$, then $\delta_{g_b}^*\omega$ satisfies the linearized gauge condition.)
\item The {\it constraint propagation wave operator}
\begin{align*}
\Box^\CD_{g_b,\gamma_C}=2\delta_{g_b}\sfG_{g_b}\delta^*_{g_b,\gamma_C}.
\end{align*} 
(If $L_b h=0$, then the gauge $1$-form $\eta=\delta_{g_b,\gamma_\Ups}\sfG_{g_b}h$ satisfies $\Box^\CD_{g_b,\gamma_C}\eta=0$.)
\end{enumerate}
The Fourier transformed operators are operators on $X_b^\circ$ defined by conjugation with $e^{-i\sigma t_*}$, so
\begin{align*}
\wh \Box^{\Upsilon}_{g_b,\gamma_{\Upsilon}}(\sigma)=e^{i\sigma t_*}\Box^{\Upsilon}_{g_b,\gamma_{\Upsilon}}e^{-i\sigma t_*},
\end{align*} 
similarly for $\wh \Box^\CD_{g_b,\gamma_C}(\sigma)$ and $\wh L_b(\sigma)$. Let us recall \cite[Theorem 3.5, Proposition 3.7, Proposition 3.12]{Hi24}:

\begin{thm}[Gauge modification and constraint damping: spectral theory]
\label{thm6.1}
 Let $\Box_{g_b}$ be one of the operators $\Box^{\Upsilon}_{g_b,\gamma_{\Upsilon}}$, $\Box^\CD_{g_b,\gamma_C}$.   
 \begin{enumerate}
  \item For $\Im\sigma\geq 0$, $\sigma\neq 0$, the operator
    \[
      \wh{\Box_{g_b}}(\sigma)\colon\left\{\omega\in \Hbext^{s,\ell}(X;\wt\Tsc{}^*X)\colon\wh{\Box_g}(\sigma)\omega\in\Hbext^{s,\ell+1}(X;\wt\Tsc{}^*X)\right\} \to\Hbext^{s,\ell+1}(X;\wt\Tsc{}^*X)
    \]
    is invertible when $s>\frac32$, $\ell<-\half$, $s+\ell>-\half$.
  \item For $s>\frac32$ and $\ell\in(-\tfrac32,-\half)$, the operator
    \begin{equation}
    \label{Eq1StatOp}
    \begin{split}
      \wh{\Box_{g_b}}(0)\colon&\left\{\omega\in\Hbext^{s,\ell}(X;\wt\Tsc{}^*X)\colon\wh{\Box_g}(0)\omega\in\Hbext^{s-1,\ell+2}(X;\wt\Tsc{}^*X)\right\} \\
        &\qquad\to\Hbext^{s-1,\ell+2}(X;\wt\Tsc{}^*X)
    \end{split}
    \end{equation}
    is a Fredholm operator of index $0$. 
  \end{enumerate}
  Moreover there exist $\cd_{C}, \, \cd_{\Upsilon}\in C_c(X_b^\circ;T^*_{X_b^\circ}M^\circ_b)$ such that for sufficiently small $\gamma_C,\, \gamma_{\Upsilon}$ the above operators are invertible on the stated spaces for all ${\rm Im}\, \sigma\ge 0$. 
  \end{thm}

We recall that we use the scattering density (i.e.\ the spatial part of the metric density $|\dd g_b|$) throughout. The spaces used for $\sigma\neq 0$ in Theorem~\ref{thm6.1} and \cite{Hi24} are discussed in the proof of Theorem~\ref{ThmOp} below.

We will fix in the following $\cd_{C}, \, \cd_{\Upsilon}\in C_c(X_b^\circ;T^*_{X_b^\circ}M^\circ_b)$ and $\gamma_C,\, \gamma_{\Upsilon}$ such that $\wh \Box^\CD_{g_b,\gamma_C}(\sigma)$, $\wh  \Box^{\Upsilon}_{g_b,\gamma_{\Upsilon}}(\sigma)$ are invertible on the above spaces for all ${\rm Im}\, \sigma\ge 0$. We note that $\gamma_C,\, \gamma_{\Upsilon}$ can be chosen arbitrary small; see \cite[\S\S{3.1}, {3.2}]{Hi24} for details. 
 
  \begin{thm}[Linearized gauge-fixed Einstein operator: spectral theory]
\label{ThmOp}
    Suppose that $s>\frac52$ and $\ell<-\half$ with $s+\ell>-\half$. Then for $\gamma_C,\, \gamma_{\Upsilon}$ sufficiently small we have:
  \begin{enumerate}
  \item\label{ItOpLow} {\rm (Uniform estimates for finite $\sigma$.)} For any fixed $C>1$, and $s_0<s$, $\ell_0<\ell$, there exists a constant $C'>0$ such that
    \begin{equation}
    \label{EqOpFinite}
      \Vert u\Vert_{\Hbext^{s,\ell}} \leq C'\Bigl(\bigl\Vert\wh{L_b}(\sigma)u\bigr\Vert_{\Hbext^{s-1,\ell+2}} + \Vert u\Vert_{\Hbext^{s_0,\ell_0}} \Bigr)
    \end{equation}
    for all $\sigma\in\C$, $\Im\sigma\in[0,C]$, satisfying $C^{-1}\leq |\sigma|\leq C$. If $\ell\in(-\tfrac32,-\half)$, then this estimate holds uniformly down to $\sigma=0$, i.e.\ for $|\sigma|\leq C$.
  \item\label{ItOpHigh} {\rm (High energy estimates in strips.)} For any fixed $C>0$, there exist $C_1>1$ and $C'>0$ such that for $\sigma\in\C$, $\Im\sigma\in[0,C]$, $|\Re\sigma|>C_1$, and $h:=|\sigma|^{-1}$, we have
    \begin{equation}
    \label{EqOpHigh}
      \Vert u\Vert_{\Hbhext^{s,\ell}} \leq C'\bigl\Vert\wh{L_b}(\sigma)u\bigr\Vert_{\Hbhext^{s,\ell+1}}.
    \end{equation}
  \end{enumerate}
  Moreover, the operators
  \begin{subequations}
  \begin{alignat}{3}
  \label{EqOpFredS}
    \wh{L_b}(\sigma)&\colon\{u\in\Hbext^{s,\ell}(X)\colon\wh{L_b}(\sigma)u\in\Hbext^{s,\ell+1}(X)\}&&\to\Hbext^{s,\ell+1}(X), \quad \Im\sigma\geq 0,\ \sigma\neq 0, \\
  \label{EqOpFred0}
    \wh{L_b}(0) &\colon \{u\in\Hbext^{s,\ell}(X)\colon\wh{L_b}(0)u\in\Hbext^{s-1,\ell+2}(X)\}&&\to \Hbext^{s-1,\ell+2}(X)
  \end{alignat}
  \end{subequations}
  are Fredholm operators of index $0$.
\end{thm}
\begin{proof}
The proof is essentially the same as for \cite[Theorem 4.3]{HHV21}, \cite[Theorem 7.1]{AHW22}, but, in addition, one has to take into account the modification of the gauge. This can be done following \cite[Proposition 3.15, Proposition 3.16]{Hi24}. The spaces in \cite{Hi24} are slightly different from the above choice: namely, scattering spaces are used in \cite{Hi24} at non-zero energy. The choice of these spaces really comes from \cite{Va21a}, \cite{Va21b}, where scattering-b-Sobolev spaces are used. The spaces in \cite{Hi24} as well as the above spaces then correspond to specializations of these scattering-b-estimates for particular choices of the scattering- and b-regularity and decay orders. We refer to the proof of \cite[Theorem~4.7]{AHW22} for details. The threshold regularity is computed in the appendix to be $\frac52$ for the unmodified operator (initial wave map gauge and usual implementation). Once $\cd_{C}$, $\cd_{\Upsilon}\in C_c(X_b^\circ;T^*_{X_b^\circ}M^\circ_b)$ fixed, the threshold regularity depends continuously on $\gamma_C$ and $\gamma_{\Upsilon}$, which gives the result. 
\end{proof}

\section{Mode analysis}
\label{Sec7}

We next study the (zero) modes for the linearized Einstein operator (with and without gauge) as well as for the two $1$-form wave operators $\Box_{g_b,\gamma^\Ups}^\Ups$, $\Box^\CD_{g_b,\gamma_C}$. We do this by combining the results of \cite{Hi24} and \cite{AHW22}. Let us first recall from \cite[Lemma 3.14]{Hi24} that the linearized Kerr metric can be put into the chosen gauge, now that we have eliminated the zero mode of $\Box_{g_b,\gamma^\Ups}^\Ups$:

\begin{lemma}[Gauged linearized Kerr metrics]
There exists $1$-forms $\omega(\dot{b})\in \cA^{0-}(X_b;T^*_{X_b}M_b)$, depending linearly on $\dot{b}\in \R^4$, such that
\begin{align*}
  g_b^{\prime\Upsilon}(\dot{b}):=\dot g_b(\dot{b})+\delta_{g_b}^*\omega(\dot{b})\in (\ker\delta_{g_b,\gamma_\Ups}\sfG_{g_b}) \cap (\rho \cC^{\infty}+\cA^{2-}).
\end{align*}  
\end{lemma}

In the following some of the mode solutions will be parametrized by scalar/vector type spherical harmonics $\scalspace_l$ and $\vectspace_l$. If $Y_{l m}\in\CI(\Sph^2)$ are the usual spherical harmonics, then
\begin{align*}
\scalspace_l := \mathspan\{Y_{l m} \colon |m|\leq l\},\quad \vectspace_l := \slstar\sld\scalspace_l \quad (l\geq 1).
\end{align*}
The slashes indicate that these are operators on the unit 2-sphere, so for example $\sld\colon\CI(\Sph^2)\to\CI(\Sph^2;T^*\Sph^2)$ is the exterior differential and $\slstar$ is the Hodge-star. We refer to \cite[\S5]{HHV21} for details.

\subsection{Mode solutions for the \texorpdfstring{$1$}{1}-form wave operators}

To describe the kernel of the gauge-fixed linearized Einstein operator, we first have to go back to the kernel of $\wh \Box^{\Upsilon}_{g_b,\gamma_{\Upsilon}}(0)$ in weaker spaces. The following result, constructing asymptotic translations as gauge potentials, is close to \cite[Lemmas~3.11 and~3.13]{Hi24}. 

\begin{lemma}[Asymptotic translations]
Let $\chi=\chi(r)\in\CI(\R)$ be equal to $0$ in $r\le 3\bhm$ and $1$ in $r\ge 4\bhm$. Then
\begin{align}
\label{7.1}
\ker\wh{\Box_{g_b,\gamma_{\Upsilon}}^{\Upsilon}}(0)\cap \Hbext^{\infty,-\frac32-}=\la \omega_{{\rm s}0}^{(0)}\ra\oplus\{\omega_{{\rm s}1}(\scal):\, \scal\in \scalspace_1\},\\
\label{7.2}
\ker\wh{\Box_{g_b,\gamma_{C}}^\CD}(0)^*\cap \Hbsupp^{-\infty,-\frac32-}=\la \omega_{{\rm s}0}^*\ra\oplus\{\omega_{{\rm s}1}^*(\scal):\, \scal\in \scalspace_1\},
\end{align}
where $V^{\flat}:=g_b(V,\cdot)$ and
\begin{alignat*}{2}
\omega_{{\rm s}0}^{(0)}&=\partial_t^{\flat}, &\quad \omega_{{\rm s}1}(\scal)-\dd(r \scal)&\in \cA^{1-},\\
\omega_{{\rm s}0}^*-\chi \partial_t^\flat &\in\Hbsupp^{-\infty,-\frac12-}, &\quad \omega_{{\rm s}1}^*(\scal)-\chi \dd(r\scal)&\in \Hbsupp^{-\infty,-\frac12-}.  
\end{alignat*}
\end{lemma}
\begin{proof}
  The proof follows \cite[Lemma 3.11]{Hi24} closely; this shows that the right hand side of the equations~\eqref{7.1} and \eqref{7.2} is contained in the left hand side. For the reverse inclusion, we follow \cite[Proposition~5.2]{AHW22}. Note that the operators considered here are slightly different from those considered in \cite{AHW22}; however, the argument we need to repeat here from the proof of \cite[Proposition~5.2]{AHW22} is a normal operator argument, and the relevant normal operators here and in \cite{AHW22} are the same.
\end{proof}

We now define $h_{{\rm s}1}(\scal):=\delta_{g_b}^*\omega_{{\rm s}1}(\scal)\in\Hbext^{\infty,\frac12-}$. 
\begin{lemma}[An auxiliary gauge potential]
\label{Lemma10Gens12}
  There exists
  \[
    \omega_{{\rm s}1}^{(1)}(\scal) \in \ker\wh{\Box_{g_b,\gamma_{\Upsilon}}^{\Upsilon}}(0) \cap \Hbext^{\infty,-\frac52-},\quad
    \omega_{{\rm s}1}^{(1)}(\scal)-r\scal\,\dd \ft\in\Hbext^{\infty,-\frac32-},
  \]
  depending linearly on $\scal\in\scalspace_1$.
\end{lemma}

In the proof one corrects $r\scal\, \dd \ft$ to an actual solution by normal operator arguments. It works in the same way as in the small $a$ case;\footnote{It is even easier because $\ker( \wh \Box^{\Upsilon}_{g_b,\gamma_{\Upsilon}})^*\cap \Hbsupp^{-\infty,-\frac12+}=\{0\}$ in the present case} see \cite[Lemma 7.11]{HHV21} for details. 

\begin{lemma}[No generalized zero modes with spatial decay]
  For, $d\ge 0$ there does not exist a $1$-form of the form $\omega=\sum_{j=0}^d \omega_j t_*^j$ with $\omega_j\in \Hbext^{\infty,-\frac32+},\, \omega_d\neq 0$ such that $\wh{\Box_{g_b,\gamma_{\gamma_{\Upsilon}}}^{\Upsilon}}(0)\omega=0$. 
\end{lemma}
\begin{proof}
  As $\partial_{t_*}$ is Killing we see that $\omega_d\in \Hbext^{\infty,-\frac32+}$ has to solve $\wh{\Box_{g_b,\gamma_{\Upsilon}}^{\Upsilon}}(0)\omega_d=0$; it is thus zero by Theorem \ref{thm6.1}. Repeating the argument shows that $\omega_j=0$ for all $j$. 
\end{proof}

\begin{prop}[Asymptotic Lorentz boosts as gauge potentials]
\label{prop7.5}
There exists
\begin{align*}
\hat{\omega}_{{\rm s}1}(\scal)\in \ker\Box_{g_b,\gamma_{\Upsilon}}^{\Upsilon}\cap\Poly^1(t_*)\Hbext^{\infty,-\frac52-}
\end{align*}
with linear dependence on $\scal\in \scalspace_1$, which in addition satisfies 
\begin{align*}
\delta_{g_b}^*\wh \omega_{{\rm s}1}\in\Poly^1(t_*)\Hbext^{\infty,-\frac12-}. 
\end{align*}
We have  $\hat{\omega}_{{\rm s}1}(\scal)=t(\dd x^i +\cA^1)-x^i\,\dd t+\cA^{0-}$ when $\scal=r^{-1}x^i$. 
\end{prop}
\begin{proof}
We follow the proof of \cite[Proposition 7.13]{HHV21}. Fixing $\scal\in\scalspace_1$ we make the ansatz 
\begin{align*}
\hat{\omega}_{{\rm s}1}(\scal)=\ft\omega_{{\rm s}1}(\scal)+\breve{\omega}_{{\rm s}1}
\end{align*}
and write 
\begin{align*}
\breve \omega_{{\rm s}1}=(-1)\omega_{{\rm s}1}^{(1)}(\scal)+\breve \omega'_{{\rm s}1}
\end{align*}
with $\breve\omega'_{{\rm s}1}$ to be determined. This then gives:
\[
    \hat h_{{\rm s}1} := \delta_{g_b}^*\hat\omega_{{\rm s}1} = \ft h_{{\rm s}1}(\scal) + \bigl([\delta_{g_b}^*,\ft]\omega_{{\rm s}1}(\scal) + \delta_{g_b}^*\breve\omega_{{\rm s}1}\bigr).
  \]
Using $h_{{\rm s}1}(\scal)\in \Hbext^{\infty,\frac12-}$ and $t_*-\ft\in\cA^{-1}$, we therefore have $\hat\omega_{{\rm s}1}\in\ker\Box^{\Upsilon}_{g_b,\gamma_{\Upsilon}}$, $\hat h_{{\rm s}1}\in\Poly^1(t_*)\Hbext^{\infty,-\frac12-}$ provided the two conditions
  \begin{subequations}
  \begin{gather}
  \label{EqL0Lins1ExistBox}
    \wh{\Box^{\Upsilon}_{g_b,\gamma_{\Upsilon}}}(0)\breve\omega_{{\rm s}1}=-[\Box^{\Upsilon}_{g_b,\gamma_{\Upsilon}},\ft]\omega_{{\rm s}1}(\scal), \\
  \label{EqL0Lins1ExistDel}
    [\delta_{g_b}^*,\ft]\omega_{{\rm s}1}(\scal) + \delta_{g_b}^*\breve\omega_{{\rm s}1} \in \Hbext^{\infty,-\frac12-}
  \end{gather}
  \end{subequations}
  hold. The membership~\eqref{EqL0Lins1ExistDel} holds provided $\breve\omega'_{{\rm s}1}\in \Hbext^{\infty,-\frac32-}$ because 
  \begin{align*}
[\delta_{g_b},\ft]\dd(r\scal)+\delta_{g_b}^*\omega_{{\rm s}1}^{(1)}\in \Hbext^{\infty,-\frac12-},
  \end{align*}
  see \cite[Proof of Proposition 7.13]{HHV21} for details. \eqref{EqL0Lins1ExistBox} is equivalent to 
\begin{align}
\label{7.4}
\wh{\Box_{g_b,\gamma_{\Upsilon}}^{\Upsilon}}(0)\breve{\omega}'_{{\rm s}1}=-[\Box_{g_b,\gamma_{\Upsilon}}^{\Upsilon},\ft]\omega_{{\rm s}1}(\scal).
\end{align}
We expand the commutator $[\Box^{\Upsilon}_{g_b,\gamma_{\Upsilon}},\ft]=[2\delta_{g_b,\gamma_{\Upsilon}}\sfG_{g_b},\ft]\delta_{g_b}^*+2\delta_{g_b,\gamma_{\Upsilon}}\sfG_{g_b}[\delta_{g_b}^*,\ft]$, and find
  \begin{align*}
    \wh{\Box_{g_b,\gamma_\Ups}^\Ups}(0)\breve\omega_{{\rm s}1}' &= - [2\delta_{g_b,\gamma_{\Upsilon}}\sfG_{g_b},\ft]h_{{\rm s}1}(\scal) \\
      &\qquad - 2\delta_{g_b,\gamma_{\Upsilon}}\sfG_{g_b}\bigl([\delta_{g_b}^*,\ft]\omega_{{\rm s}1}(\scal)-\delta_{g_b}^*\omega_{{\rm s}1}^{(1)}(\scal)\bigr) \in \Hbext^{\infty,\frac12-}.
  \end{align*}
We can solve this equation for $\breve{\omega}'_{{\rm s}1}\in \Hbext^{\infty,-\frac32-}$ because 
\[
  \ker\wh{\Box_{g_b,\gamma_{\Upsilon}}^{\Upsilon}}(0)^*\cap \Hbsupp^{-\infty,-\frac12+}\subset \ker\wh{\Box_{g_b,\gamma_{\Upsilon}}^{\Upsilon}}(0)^*\cap \Hbsupp^{-\infty,-\frac12-}=\{0\}.\qedhere
\]
\end{proof}

\subsection{Mode solutions for the gauge-fixed linearized Einstein operator} 

From \cite[Proposition 3.16]{Hi24}, we get:

\begin{thm}[Zero energy mode solutions for $L_b$]
\label{ThmL0}
Let $s>\frac52$. 
    \begin{enumerate}
  \item For $\Im\sigma\geq 0$, $\sigma\neq 0$, the operator
    \begin{align*}
      \wh{L_{b}}(\sigma)&\colon\{h\in \Hbext^{s,\ell}(X;S^2\,\wt\Tsc{}^*X)\colon\wh{L_{b}}(\sigma)h\in\Hbext^{s,\ell+1}(X;S^2\,\wt\Tsc{}^*X)\} \\
        &\qquad \to\Hbext^{s,\ell+1}(X;S^2\,\wt\Tsc{}^*X)
    \end{align*}
    is invertible when $\ell<-\half$, $s+\ell>-\half$.
  \item For $\ell\in(-\tfrac32,-\half)$, the zero energy operator
    \begin{equation}
    \label{EqL0Op}
    \begin{split}
      \wh{L_{b}}(0)&\colon\{h\in\Hbext^{s,\ell}(X;S^2\,\wt\Tsc{}^*X)\colon\wh{L_{b}}(0)h\in\Hbext^{s-1,\ell+2}(X;S^2\,\wt\Tsc{}^*X)\} \\
        &\qquad \to\Hbext^{s-1,\ell+2}(X;S^2\,\wt\Tsc{}^*X)
    \end{split}
    \end{equation}
    has 7-dimensional kernel and cokernel, given by 
    \begin{alignat*}{2}
    \cK_b&:=\ker_{\Hbext^{\infty,-\frac12-}}\hat{L}_b(0)&&=\mathspan\left(\{g_b^{\prime\Ups}(\dot{b}):\, \dot{b}\in \R^4\}\cup \{h_{{\rm s}1}(\scal):\, \scal \in \scalspace_1\}\right), \\ 
    \cK_b^*&:= \ker_{\Hbsupp^{-\infty,-\frac12-}}\hat{L}_b(0)^*&&=\mathspan\bigl(\{h_{{\rm s}0}^*\}\cup\{h_{{\rm v}1}^*(\vect):\vect\in \vectspace_1\}\cup\{h^*_{{\rm s}1}(\scal):\scal\in \scalspace_1\}\bigr).   
    \end{alignat*} 
   \end{enumerate}
   Moreover, we have
   \begin{align*}
   h_{{\rm s}1}(\scal)\in \Hbext^{\infty,\frac12-}, \quad h_{{\rm s}0}^*,\, h_{{\rm s}1}(\scal),\, h_{{\rm v}1}(\vect)\in \Hbsupp^{-\infty,\frac12-}. 
   \end{align*}
\end{thm}
We refer to \cite{Hi24} for the detailed construction of $h_{{\rm s}0}^*$, $h_{{\rm v}1}^*(\vect)$, $h^*_{{\rm s}1}(\scal)$.

\subsection{Generalized mode solutions}

Let 
\begin{equation}
\label{EqBreveh}
  \breve{h}_{{\rm s}1}(\scal):=\delta_{g_b}^*\breve{\omega}_{{\rm s}1}+[\delta_{g_b}^*,t_*]\omega_{{\rm s}1}(\scal) \in \Hbext^{\infty,-\frac12-} = \cA^{1-};
\end{equation}
thus, $\breve{h}_{b,{\rm s}1}$ is a solution of 
\begin{align}
\label{7.3.1}
[L_b,t_*]h_{{\rm s}1}=-L_b\breve{h}_{{\rm s}1}\Leftrightarrow L_b(t_*h_{{\rm s}1}(\scal)+\breve{h}_{{\rm s}1})=0. 
\end{align}
The (normal operator based) arguments in the proof of \cite[Lemma~9.6]{HHV21} apply verbatim and give the more precise description $\breve{h}_{{\rm s}1}(\scal)\in\rho\CI+\cA^{2-}$.

To analyze generalized mode solutions of $L_b$, we need to recall the evaluation of certain pairings which were previously computed in \cite[Lemma~3.17]{Hi24}:

\begin{lemma}[$L^2$-pairings]
\label{Hi3.17}
For ${\bhq}\in \R^3$ we write 
\begin{align*}
\vect(\bfq)=\left(\bfq\times \frac{x}{\vert x\vert}\right)\cdot \frac{\dd x}{\vert x\vert}\in \vectspace_1,\quad \scal(\bfq)=\bfq \cdot \frac{x}{\vert x\vert}\in \scalspace_1.
\end{align*}
\begin{enumerate}
\item\label{ItHi3.17Mass} (Mass changes). We have $\la [L_b,t_*]g^{\prime\Upsilon}_b(\dot{\bhm},0),h^*\ra=-16\pi \dot{\bhm}$ (for $h^*=h_{{\rm s}0}^*$), resp. $-16\pi(\bhq \cdot \bha)\dot{\bhm}$ (for $h^*=h_{{\rm v}1}(\vect(\bhq)$), resp. $0$ (for $h^*=h_{{\rm s}1}(\scal),\, \scal\in \scalspace_1$). 
\item\label{ItHi3.17Ang} (Angular momentum changes). We have $\la [L_b,t_*]g^{\prime\Ups}_b(0,\dot{\bha}),h^*\ra=0$ (for $h^*=h_{{\rm s}0}^*$), resp. $-16\pi\bhm(\bhq \cdot \dot{\bha})$ (for $h^*=h^*_{{\rm v}1}(\vect(\bhq)$), resp. $0$ (for $h^*=h^*_{{\rm s}1}(\scal),\, \scal\in \scalspace_1$). 
\item\label{ItHi3.17COM} (Center of mass changes). We have $\la \frac{1}{2}[[L_b,t_*],t_*]h_{{\rm s}1}(\scal(c))+[L,t_*]\breve{h}_{{\rm s}1}(\scal(c)),h^*\ra=0$ (for $h^*=h_{{\rm s}0}^*,\, h^*_{{\rm v}1}(\vect),\, \vect\in \vectspace_1)$), resp. $8\pi\bhm(\bhq\cdot c)$ (for $h^*=h^*_{{\rm s}1}(\scal(\bhq))$). 
\end{enumerate}
\end{lemma}

We moreover recall the identity
\begin{equation}
\label{EqHi3.17Extra}
  \la [L_b,t_*]h_{{\rm s}1}(\scal),h_{{\rm s}0}^*\ra = 0.
\end{equation}
This follows from $[L_b,t_*]h_{{\rm s}1}=-\wh{L_b}(0)\breve h_{{\rm s}1}$ and integration by parts.

Let now $\wh{\cK}_b:=\ker L_b\cap {\rm Poly}^1(t_*)\Hbext^{\infty,-\frac12-}$. The proof of the following result is simpler than that of the analogous \cite[Proposition~9.4]{HHV21}:

\begin{prop}[Generalized zero modes of $L_b$]
\label{prop6.4}
  Set $\hat{h}_{{\rm s}1}(\scal)=\delta_{g_b}^*\hat{\omega}_{{\rm s}1}(\scal)$. Then
  \begin{align*}
 \wh{\cK}_b=\mathspan\bigl(\cK_b\cup \{\hat{h}_{{\rm s}1}(\scal) \colon \scal\in \scalspace_1 \}\bigr). 
  \end{align*}
\end{prop}
\begin{proof}
     The inclusion `$\supseteq$' is clear. For the other inclusion, let $h=t_*h_1+h_0\in \wh{\cK}_b$. Then $L_bh_1=0$ and thus $h_1=g_b^{\prime\Ups}(\dot{\bhm},\dot{\bha})+h_{{\rm s}1}(\scal)$ for some $(\dot{\bhm},\dot{\bha})\in \R^4$, $\scal\in \scalspace_1$. We therefore have 
  \begin{align*}
  -[L_b,t_*]h_1=L_bh_0.
  \end{align*}
  Pairing the equation with $h^*_{{\rm s}0}$ and using Lemma~\ref{Hi3.17} as well as the identity~\eqref{EqHi3.17Extra} gives $\dot{m}=0$. Pairing with $h^*_{b,v1}(V(\dot{\bha}))$  and using Lemma \ref{Hi3.17}\eqref{ItHi3.17Ang} gives $\dot{\bha}=0$. Thus, $h_1=h_{{\rm s}1}(\scal)=\delta_{b}^*\omega_{{\rm s}1}$. Now,
  \begin{align*}
  h=\delta_{g_b}^*(t_*\omega_{{\rm s}1})-[\delta_{g_b}^*,t_*]\omega_{{\rm s}1}+h_0. 
  \end{align*}
  Thus, $h-\hat{h}_{{\rm s}1}$ is a stationary solution of $L_b(h-\hat h_{{\rm s}1})=0$, so $h-\hat h_{{\rm s}1}\in\cK_b$.
\end{proof}

\begin{lemma}[No generalized zero modes of $L_b$ with at least quadratic growth]
  \label{lemma7.8}
Let $d\ge 2$. There does not exist $h=\sum_{j=0}^d t_*^jh_j$ with $h_j\in \Hbext^{\infty,-\frac12-}$ with $h_d\neq 0$ for which $L_bh=0$. 
\end{lemma}
\begin{proof}
We can suppose $d=2$. Indeed, as $\partial_{t_*}$ is Killing, $\partial_{t_*}^{d-2}h$ is also solution. The same argument shows that $2t_*h_2+h_1$ is also solution, thus Proposition~\ref{prop6.4} shows that $h_2=h_{{\rm s}1}(\scal)$ where $\scal\neq 0$. Then $h_1$ has to solve the equation $2[L_b,t_*]h_{{\rm s}1}(\scal)=-L_b h_1$. Thus, $h_1=2\breve{h}_{{\rm s}1}(\scal)+h'$ for some $h'\in\cK_b$. Then $h_0$ has to solve 
\begin{align*}
[[L_b,t_*],t_*]h_{{\rm s}1}(\scal)+[L_b,t_*](2\breve{h}_{{\rm s}1}(\scal)+h')=-L_b h_0.
\end{align*}
By Lemma \ref{Hi3.17}\eqref{ItHi3.17COM}, the pairing of the left hand side of this equation with $h^*_{{\rm s}1}(\scal)$ is non-zero, contradicting the existence of $h_0$.
\end{proof}

\begin{rmk}[Differences to \cite{HHV21}]
\label{RmkDiff}
  Due to the different choice of gauge which eliminates zero modes of $\Box_{g_b,\gamma^\Ups}^\Ups$, the mode solutions $\delta_{g_b}^*\omega_{b,{\rm s}0}$ in \cite{HHV21} coming from Coulomb type solutions of the $1$-form wave operator are now eliminated. Moreover, linearizations of the Kerr metric in the mass parameter cannot be directly put in the gauge of \cite{HHV21}; instead, they show up there only as generalized mode solutions. In our gauge, they do exist as stationary mode solutions.   
\end{rmk}

Let now $h_{b,{\rm s}0}:=g_b^{\prime\Ups}(1,0)$, $h_{b,{\rm v}1}(\vect):=g^{\prime\Ups}_{b}(0,\vect)$. We define the spaces
\begin{equation}
\label{EqRZeroModes}
\begin{aligned}
  \cK_{b,{\rm s}0} &:= \C h_{b,{\rm s}0}, && \quad & \cK_{b,{\rm s}0}^* & := \C h_{b,{\rm s}0}^*, \\
  \cK_{b,v} &:= h_{b,{\rm v}1}(\vect_1), & && \cK_{b,v}^* & := h_{b,{\rm v}1}^*(\vect_1), \\
  \cK_{b,t} &:= h_{b,{\rm s}1}(\scal_1), && \quad & \cK_{b,t}^* & := h_{b,t}^*(\scal_1), \\
  \cK_{b,l} &:= \cK_{b,{\rm s}0}\oplus\cK_{b,v},&& \quad & \cK_{b,l}^* &:= \cK_{b,{\rm s}0}^* \oplus \cK_{b,v }^*, \\
    \cK_b &:= \cK_{b,t}\oplus\cK_{b,l}, & && \cK_b^*&:=\cK_{b,t}^*\oplus\cK_{b,l}^*.
\end{aligned}
\end{equation}
(Here the subscript $v$ stands for `vector type,' $t$ stands for `translations,' and $l$ stands for `linearized Kerr.') Note that the splitting is slightly different from the one in \cite[\S{11.1}]{HHV21}. In the setting of the present paper, $\wh{L_b}(\sigma)$ is expected to be more singular on $\cK_{b,t}$ than on $\cK_{b,l}$ because of the existence of linearly growing modes with leading term in $\cK_{b,t}$. Also, the gauge in \cite{HHV21} did not allow to capture $g^{\prime\Ups}_b(1,0)$, cf.\ Remark~\ref{RmkDiff}.

\section{Structure and regularity of the resolvent of the gauge-fixed operator}
\label{Sec9}

We now explain the (minor) necessary changes in the analysis of the resolvent $\wh{L_b}(\sigma)^{-1}$ as compared to \cite{HHV21} and \cite{HaefnerHintzVasyKerrErratum}.

\subsection{Structure of the resolvent}

We fix
\begin{align}
\label{9.1}
s>\tfrac{5}{2},\quad \ell\in (-\tfrac{3}{2},-\tfrac{1}{2}). 
\end{align} 
The precise decay estimates for $L_b$ on spacetime are linked to the regularity of the resolvent $\wh{L_b}(\sigma)^{-1}$ at zero energy. Because of the existence of (generalized) mode solutions, there exists a non-trivial kernel at $\sigma=0$ and we have to separate this corresponding singular part from a more regular part. Another important aspect is that the domains 
\[{\mathcal X}_b^{s,\ell}:=\{u\in\Hbext^{s,\ell}(X)\colon\wh{L_b}(\sigma)u\in\Hbext^{s,\ell+1}(X)\}\]
depend on $\sigma$. To circumvent these difficulties we use:

\begin{lemma}[Modification of the spectral family]
\label{LemmaRPert0}
  There exist $V\in\Psi^{-\infty}(X^\circ;S^2 T^*X_b^\circ)$, with compactly supported Schwartz kernel, and a constant $C_1>0$ such that
  \[
    \check L_b(\sigma) := \wh{L_b}(\sigma) + V \colon \cX_b^{s,\ell}(\sigma)\to\Hbext^{s-1,\ell+2}
  \]
  is invertible for $\sigma\in\C$, $\Im\sigma\geq 0$, $|\sigma|<C_1$. Moreover, $\check L_b(\sigma)^{-1}$ is continuous in $\sigma$ with values in $\cL_{\rm weak}(\Hbext^{s-1,\ell+2},\Hbext^{s,\ell})\cap\cL_{\rm op}(\Hbext^{s-1+\eps,\ell+2+\eps},\Hbext^{s-\eps,\ell-\eps})$, $\eps>0$.
\end{lemma}
\begin{proof}
The proof is analogous to the one of \cite[Lemma 11.1]{HHV21}. The only difference is that in \cite{HHV21} it was first carried out for $L_{b_0}$, $b_0=(\bhm,0)$, and then perturbation arguments were used. Thanks to the results of \S\ref{Sec7} we now have all relevant information for general subextremal $b$ and can now carry out the construction directly for $L_b$. 
\end{proof}
    
For $\check L_b(\sigma)$ as in Lemma~\ref{LemmaRPert0}, we set
\begin{equation}
\label{EqRwtcK}
\begin{split}
  \wt\cK_{b,l} &:= \check L_b(0)\cK_{b,l}, \\
  \wt\cK_{b,t} &:= \check L_b(0)\cK_{b,t}, \\
  \wt\cK_b &:= \wt\cK_{b,t}\oplus\wt\cK_{b,l}.
\end{split}
\end{equation}
By definition of $\check L_b(0)$, these are subspaces of $\CIc(X^\circ_b;S^2 T^*X_b^\circ)$. We fix a complementary subspace $\wt\cK_b^\perp\subset\Hbext^{s-1,\ell+2}$ of $\wt\cK_b$.

We want to decompose the target space $\Hbext^{s-1,\ell+2}$ into the range of $\wh{L_b}(0)$ and a complement. To do this, we use a slight generalization of the procedure used in the proof of Lemma~\ref{LemmaRPert0}; the following lemma is proved as in \cite[Lemma~11.2]{HHV21}.

\begin{lemma}[Projection off the range]
\label{LemmaRProj}
  There exists a linear projection map $\Pi_b^\perp\colon\Hbext^{s-1,\ell+2}\to\Hbext^{s-1,\ell+2}$ which is of rank $7$ and satisfies
  \[
    \la(I-\Pi_b^\perp)f,h^*\ra=0\quad\forall\,h^*\in\cK_b^*.
  \]
  The Schwartz kernel of $\Pi_b^\perp$ can be chosen to be independent of $s,\ell$ satisfying~\eqref{9.1}.
\end{lemma}

Defining the complementary projection
\[
  \Pi_b := I-\Pi_b^\perp \colon \Hbext^{s-1,\ell+2} \to \ran\Pi_b=\ann\cK_b^*=\ran_{\cX_b^{s,\ell}(0)}\wh{L_b}(0),
\]
we then split domain and target of $\wh{L_b}(0)\check L_b(0)^{-1}$ according to
\begin{equation}
\begin{aligned}
\label{EqRSplit}
  &\text{domain:}\ && \Hbext^{s-1,\ell+2} \cong \wt\cK^\perp \oplus \wt\cK_{b,t} \oplus \wt\cK_{b,l}, \\
  &\text{target:}\ && \Hbext^{s-1,\ell+2} \cong \ran\Pi_b \oplus \cR_t^\perp \oplus \cR_l^\perp,
\end{aligned}
\end{equation}
where $\cR_t^\perp$, resp.\ $\cR_l^\perp$, is a space of dimension $\dim\cK_{b,t}=3$, resp.\ $\dim\cK_{b,l}=4$, chosen such that the $L^2$-pairing $\cR_t^\perp\times\cK_{b,t}^*\to\C$, resp.\ $\cR_l^\perp\times\cK_{b,l}^*\to\C$, is non-degenerate. (We can choose $\cR^\perp_{t/l}$ to be a subspace of $\CIc(X^\circ;S^2 T^*X^\circ)$.) Via these pairings, we can identify
\[
  \cR_t^\perp \cong (\cK_{b,t}^*)^*,\quad
  \cR_l^\perp \cong (\cK_{b,l}^*)^*;
\]
we shall use these identifications implicitly below.

The resolvent at $\sigma\neq 0$, $\Im\sigma\geq 0$, exists: 
\begin{prop}[Resolvent near $0$]
\label{PropRExist}
 The operator $\wh{L_b}(\sigma)\colon\cX^{s,\ell}_b(\sigma)\to\Hbext^{s-1,\ell+2}$ is invertible for $\sigma\in\C$, $\Im\sigma\geq 0$, $\sigma\neq 0$. The inverse is given by 
 \begin{equation}
  \label{9.4}
    \wh{L_b}(\sigma)^{-1} :=\check L_b(\sigma)^{-1}\begin{pmatrix}
         \wt R_{0 0} & \sigma^{-1}\wt R_{0 1} & \wt R_{0 2} \\
         \sigma^{-1}\wt R_{1 0} & \sigma^{-2}\wt R_{1 1} & \sigma^{-1}\wt R_{1 2} \\
         \wt R_{2 0} & \sigma^{-1}\wt R_{2 1} & \sigma^{-1}\wt R_{2 2}
       \end{pmatrix},
  \end{equation}
where $\tilde{R}_{ij}(\sigma)$ are continuous in $\sigma$ down to zero. 
\end{prop}
\begin{proof}
The proof is analogous to that of \cite[Proposition~11.3]{HHV21}. The general strategy is to write, in the splittings~\eqref{EqRSplit} of $\Hbext^{s-1,\ell+2}$,
  \begin{equation}
  \label{EqRMtx}
    \wh{L_b}(\sigma)\check L_b(\sigma)^{-1} = \begin{pmatrix} L_{0 0} & L_{0 1} & L_{0 2} \\ L_{1 0} & L_{1 1} & L_{1 2} \\ L_{2 0} & L_{2 1} & L_{2 2} \end{pmatrix},\quad
    L_{i j}=L_{i j}(\sigma).
  \end{equation}
  For example, $L_{0 0}(\sigma):=\Pi_b\wh{L_b}(\sigma)\check L_b(\sigma)|_{\wt\cK^\perp}$. (We drop the dependence of $L_{i j}$ on $b$ from the notation.) Since the range of $\wh{L_b}(0)$ is annihilated by $\cK_b^*$, we have $L_{1 0}(0)=L_{2 0}(0)=0$; likewise, $\wh{L_b}(0)|_{\cK_b}=0$ implies that
  \[
    L_{\ker} := \begin{pmatrix} L_{0 1} & L_{0 2} \\ L_{1 1} & L_{1 2} \\ L_{2 1} & L_{2 2} \end{pmatrix}
  \]
  satisfies $L_{\ker}(0)=0$. Lastly, $L_{0 0}(0)$ is invertible. One then writes down partial Taylor expansions for the $L_{i j}$ and computes the first non-vanishing terms. This can be done exactly as in the proof of \cite[Proposition 11.3]{HHV21}. We make a few comments:

\begin{enumerate}
\item In the large $a$ case considered here, the asymptotic behavior of the metric does not change compared to the small $a$ case. Modifications due to change of gauge and constraint damping are compactly supported. As a result, the normal operator of $\wh{L_b}(0)$ and the relevant behavior of $\pa_\sigma\wh{L_b}(0)$, $\pa_\sigma^2\wh{L_b}(0)$ do not change.
\item The coefficients $L_{10}, L_{20}$ can be treated as in \cite{HHV21}. 
\item Concerning $L_{ij},\ i,j\in \{1,2\}$, one has to compute the Taylor expansion. As can be observed from the proof of \cite[Proposition~11.3]{HHV21}, the vanishing of the first Taylor coefficients is linked to the vanishing of certain pairings. More concretely:  
\begin{enumerate}
\item If for $h\in \cK_{b,t/l}$ the bilinear form $\la[L_b,t_*]h,\cdot\ra$ is non-vanishing on $(\cK_{b,t/l})^*$, then the corresponding Taylor coefficient in front of $\sigma$ is non-vanishing.
\item If the above mentioned bilinear form does vanish on $\cK_b^*$ (which in the present paper happens only for $h\in\cK_{b,t}$), then the equation $[L_b,t_*]h=L_b\breve{h}$ has a solution.  One then considers for $h\in \cK_{b,t/l}$ the pairing
\begin{align*}
\la [[L_b,t_*],t_*]h+2[L_b,t_*]\breve{h},\cdot\ra
\end{align*} 
on $(\cK_{b,t/s})^*$. If this is non-vanishing, the corresponding term in the Taylor expansion in front of $\sigma^2$ is non-vanishing. 
\end{enumerate}

(Note that the first pairing above vanishes on $\cK_b^*$ if and only if the linearized Einstein equation has a linearly growing mode with leading order term $h$. If this is the case, the second pairing vanishes on $\cK_b^*$ if and only if there is also a quadratically growing mode with leading order term $h$.) Using Lemma \ref{Hi3.17}, one thus finds that 
\begin{align*}
L_{11}=\sigma^2\wt{L}_{11}, \ L_{12}=\sigma^2\wt{L}_{12},\ L_{21}=\sigma^2\wt{L}_{21},\ L_{22}=\sigma\wt{L}_{22}
\end{align*} 
for continuous $\wt L_{ij}$.
\end{enumerate}

Summarizing we find, as in the proof of \cite[Proposition~11.3, equation~(11.30)]{HHV21},
\begin{equation}
  \label{EqROp1}
    \wh{L_b}(\sigma)\check L_b(\sigma)^{-1}
     =\begin{pmatrix}
        L_{0 0} & \sigma\wt L_{0 1} & \sigma\wt L_{0 2} \\
        \sigma\wt L_{1 0} & \sigma^2\wt L_{1 1} & \sigma^2\wt L_{1 2} \\
        \sigma\wt L_{2 0} & \sigma^2\wt L_{2 1} & \sigma\wt L_{2 2}
      \end{pmatrix}.
\end{equation}
with sufficiently precise information on $\wt L_{ij}$ to show that the inverse exists and has the form \eqref{9.4}. 
\end{proof}

Following \cite{HHV21} we can improve \eqref{EqROp1} to
 \begin{equation*}
  \label{EqROpimpr}
    \wh{L_b}(\sigma)\check L_b(\sigma)^{-1}
     =\begin{pmatrix}
        L_{0 0} & \sigma^2 L'_{0 1} & \sigma\wt L_{0 2} \\
        \sigma\wt L_{1 0} & \sigma^2\wt L_{1 1} & \sigma^2\wt L_{1 2} \\
        \sigma\wt L_{2 0} & \sigma^2\wt L_{2 1} & \sigma\wt L_{2 2}
      \end{pmatrix}.
  \end{equation*}
As noted in \cite{HaefnerHintzVasyKerrErratum} the other improvements of \cite{HHV21} are not valid because of an error in \cite[Lemma~11.7]{HHV21}. The above improvement only uses \cite[(11.38)]{HHV21} which can be arranged (in contrast to \cite[(11.39)]{HHV21}, which cannot be arranged). This leads to an improved form of the inverse 
\begin{equation}
  \label{9.5}
    \wh{L_b}(\sigma)^{-1} :=\check L_b(\sigma)^{-1}\begin{pmatrix}
         \wt R_{0 0} & \wt R'_{0 1} & \wt R_{0 2} \\
         \sigma^{-1}\wt R_{1 0} & \sigma^{-2}\wt R_{1 1} & \sigma^{-1}\wt R_{1 2} \\
         \wt R_{2 0} & \sigma^{-1}\wt R_{2 1} & \sigma^{-1}\wt R_{2 2}
       \end{pmatrix}.
\end{equation}

\begin{definition}[Pairing]
\label{DefRBreve}
  For $h_t=h_{{\rm s}1}(\scal)\in \cK_{b,t}$, set $\breve h_t=\breve h_{{\rm s}1}(\scal)\in \cK_{b,t}$ (so that $L_b(t_* h_t+\breve h_t)=0$). We then define the non-degenerate sesquilinear pairing
  \begin{equation}
  \label{EqRBreve}
  \begin{split}
    &k_b \colon \cK_b \times \cK_b^* \to \C \\
    &\qquad k_b((h_t,g'^{\Upsilon}_b(\dot{b})),h^*) := \big\la\half\bigl([[L_b,t_*],t_*]h_t+2[L_b,t_*]\breve h_t\bigr) + [L_b,t_*]g_b'^{\Upsilon}(\dot b), h^*\big\ra.
  \end{split}
  \end{equation}
  Moreover, for $h_t^*\in \cK_{b,t}^*$ we define $\breve h_t^*$ so that $L_b(t_* h^*_t+\breve h^*_t)=0$.\footnote{The construction of $\breve h_t^*$ is completely analogous to that of $\breve h_{{\rm s}1}$ in~\eqref{EqBreveh} and Proposition~\ref{prop7.5} above.}
\end{definition}

The non-degeneracy of this pairing is a consequence of Lemma~\ref{Hi3.17}.

As in \cite[\S{11.2}]{HHV21} but with the correction from \cite{HaefnerHintzVasyKerrErratum}, we now obtain:

\begin{thm}[Precise structure of the low energy resolvent]
\label{ThmR}
  For $\Im\sigma\geq 0$, we can write
  \[
    \wh{L_b}(\sigma)^{-1} = P_b(\sigma) + \sigma^{-1}\check L_b(\sigma)^{-1}\wt R_{1 0} + L^-_b(\sigma) \colon \Hbext^{s-1,\ell+2}(X;S^2\,\wt\Tsc{}^*X) \to \Hbext^{s,\ell}(X;S^2\,\wt\Tsc{}^*X).
  \]
  The three summands are as follows:
  \begin{itemize}
  \item the regular part $L_b^-(\sigma)$ is uniformly in bounded operator norm and continuous with values in $\cL_{\rm weak}(\Hbext^{s-1,\ell+2},\Hbext^{s,\ell})\cap\cL_{\rm op}(\Hbext^{s-1+\eps,\ell+2+\eps},\Hbext^{s-\eps,\ell-\eps})$, $\eps>0$;
  \item in the second term, we regard $\wt R_{1 0}$ from~\eqref{9.5} as an operator on all of $\Hbext^{s-1,\ell+2}$ via composing with the inclusion $\wt\cK_{b,t}\hra\Hbext^{s-1,\ell+2}$ on the left and $\Pi_b$ on the right);
  \item the principal part $P_b(\sigma)$ is a quadratic polynomial in $\sigma^{-1}$ with finite rank coefficients.
  \end{itemize}
  Explicitly, $P_b(\sigma)$ is given by
  \begin{equation}
  \label{EqRPrincipal}
    P_b(\sigma)f = (-\sigma^{-2}h_t+i\sigma^{-1}\breve h_t) + i\sigma^{-1}(h'_t + g'^{\Upsilon}_b(\dot{b})),
  \end{equation}
  where $h_t,h'_t\in\cK_{b,t}$ and $g^{\prime\Ups}_b(\dot{b})\in\cK_{b,l}$ are uniquely determined by the conditions
  \begin{subequations}
  \begin{align}
  \label{EqRSing1}
    k_b((h_t, g^{\prime\Ups}_b(\dot{b})),h^*) &= \la f,h^*\ra\quad(h^*\in\cK_b^*), \\
  \label{EqRSing2}
  \begin{split}
    k_b(h'_t,h_t^*) &= -\la\half[[L_b,t_*],t_*](\breve h_t+ g^{\prime\Ups}_b(\dot{b})),h_t^*\ra - \la f,\breve h_t^*\ra \\
      &\qquad + \big\la\half\bigl([[L_b,t_*],t_*]h_t+2[L_b,t_*]\breve h_t\bigr)+[L_b,t_*] g^{\prime\Ups}_b(\dot{b}),\breve h_t^*\big\ra \quad(h_t^*\in\cK_{b,t}^*).
  \end{split}
  \end{align}
  \end{subequations}
 \end{thm}

\subsection{Regularity of the resolvent in the spectral parameter}

The arguments in~\cite[\S{12}]{HHV21} and \cite{HaefnerHintzVasyKerrErratum} apply and give:

\begin{prop}
Let $\sigma_0>0$. We have
\begin{align*}
\sigma^{-1}\check L_b(\sigma)^{-1}\wt R_{1 0} &\in H^{\frac12-\epsilon-}((-\sigma_0,\sigma_0);\cL(\Hbext^{s-1,\ell+2};\cK_{b,t})),\\
L_{b}^-&\in H^{\frac32-\epsilon-}((-\sigma_0,\sigma_0);\cL(\Hbext^{s-1,\ell+2};\Hbext^{s-\max(\epsilon,\frac12),\ell+\epsilon-1})).
\end{align*}
\end{prop}
\begin{proof}
The proof in \cite{HHV21,HaefnerHintzVasyKerrErratum} goes through essentially without any changes. The main ingredient is the structure of the resolvent and the behavior of the operator at infinity, which is given by 
\begin{equation}
\label{EqLStruct}
\begin{split}
  \wh{L_b}(\sigma) &= 2\sigma\rho(\rho D_\rho+i) + \wh{L_b}(0) + \cR_b(\sigma), \\
  &\quad \cR_b(\sigma)\in \sigma\rho^3\Diffb^1(X) + \sigma\rho^2\CI(X) + \sigma^2\rho^2\CI(X), \quad
  \wh{L_b}(0)\in\rho^2\Diffb^2(X).
\end{split}
\end{equation}
This behavior is the same in the large $a$ case, and it is unaffected by the compactly supported changes of the gauge condition and constraint damping. We do need to take into account that the elements of the different kernels have slightly changed compared to \cite{HHV21}. Let us compare regularity and decay of the different elements that are involved:
\begin{enumerate}
\item In \cite[Lemma~9.6]{HHV21} and using the notation of that paper, we have
\begin{align*}
h_{b,{\rm s}0},\ h_{b,{\rm s}1},\ h_{b,v1}\in \Hbext^{\infty,\frac12-},\\
h^*_{b,{\rm s}1}(\scal)\in \Hbsupp^{-\infty,\frac12-},\\
h_{b,v1}^*(\vect)\in \rho\cC^{\infty}+\Hbsupp^{-\infty,\frac12-},\\
\breve{h}_{b,{\rm s}0},\ \breve{h}_{b,{\rm s}1}\in \rho\cC^{\infty}+\Hbext^{\infty,\frac12-},\\
\breve{h}^*_{b,{\rm s}0},\ \breve{h}^*_{b,{\rm s}1}\in \rho\cC^{\infty}+\Hbsupp^{-\infty,\frac12-}. 
\end{align*}
\item In the present paper, we have no pure gauge ${\rm s}0$ zero energy state anymore, and also $\breve h_{b,{\rm s}0}$ (which encoded linearized mass changes) is no longer present. Instead,
\begin{align*}
g_b^{\prime\Ups}(\dot{b})\in \rho\cC^{\infty}+\cA^{2-}\subset\rho\cC^{\infty}+\Hbext^{\infty,\frac12-},\\
h_{{\rm s}1}(\scal)\in \rho^2\cC^{\infty}+\cA^{3-}\subset\Hbext^{\infty,\frac12-},\\
h_{{\rm s}0}^*,\ h_{{\rm s}1}^*(\scal),\ h^*_{{\rm v}1}(\vect)\in \Hbsupp^{-\infty,\frac12-},\\
\breve{h}_{{\rm s}1}\in \rho\cC^{\infty}+\Hbext^{\infty,\frac12-},\\
\breve{h}_{{\rm s}1}^*\in \rho \cC^{\infty}+\Hbsupp^{-\infty,\frac12-}. 
\end{align*}
\end{enumerate}
The important point is that modulo $\rho\cC^{\infty}$ the terms are in $\Hbext^{\infty,\frac12-}$ resp. $\Hbsupp^{-\infty,\frac12-}$. The leading order $\rho\cC^{\infty}$ is allowed because the leading term in $\wh L_b(\sigma)$, i.e.\ $2\sigma\rho(\rho D_\rho+i)$, annihilates it. Given this observation, the proof is unchanged. 
\end{proof}
The regularity of $L^{-}(\sigma)$ at intermediate and high frequencies is the same as in \cite{HHV21}, proofs are unchanged. We omit the details.

\section{Precise formulation of the results}
\label{Sec8}

We introduce the \emph{spacetime} Sobolev spaces
\begin{equation}
\label{EqDFn}
  \wt H_\bop^{s,\ell},\quad
  \wt H_{\bop,\rm c}^{s,\ell,k},
\end{equation}
equal to $L^2(t_*^{-1}([0,\infty));|\dd g_b|)$ for $s,\ell,k=0$. The index $\ell\in\R$ is the weight in $\rho=r^{-1}$, i.e.\ $\wt H_\bop^{s,\ell}=\rho^\ell\wt H_\bop^s$, likewise for the second (conormal) space. The index $s\in\R$ measures regularity with respect to $\pa_{t_*}$ and stationary b-vector fields on $X$. The index $k\in\N_0$ measures regularity with respect to $\la t_*\ra D_{t_*}$, so $u\in\wt H_{\bop,\rm c}^{s,\ell,k}$ if and only $(\la t_*\ra D_{t_*})^j u\in\wt H_\bop^{s,\ell}$, $j=0,\ldots,k$. We require all elements of these spaces to be supported in $t_*\geq -T$ for some fixed $T\in\R$. We write
\[
  \Hb^s(\R_{t_*}),\quad H_{\bop,{\rm c}}^{s,k}(\R_{t_*})
\]
for the analogous spaces of functions of $t_*$ only.

We define
\[
  \Sigma_0^\circ := \ft^{-1}(0),
\]
which is a (spacelike with respect to $g_b$) Cauchy surface. Identifying $\Sigma_0^\circ$ with the subset $\{r>r_0\}$ (recalling~\eqref{eqr0}) of $\R^3$, we can compactify $\Sigma_0^\circ$ at infinity to the manifold (with two boundary components) $\Sigma_0$. We recall from~\S\ref{Sec4} the notation $\wt{\Tsc^*}\Sigma_0=\Tsc^*\Sigma_0 \oplus \ul\R\,\dd\ft$ for the spacetime scattering cotangent bundle, spanned over $\CI(\Sigma_0)$ by $\dd t,\dd x^1,\dd x^2,\dd x^3$, with $(x^1,x^2,x^3)$ denoting standard coordinates on $\R^3$.

We now state several versions of our main theorem. The first concerns the forcing problem.  
\begin{thm}
\label{thm7.1}
 Let $\ell\in (-\frac{3}{2},-\frac{1}{2})$, $\epsilon\in (0,1)$, $\ell+\epsilon\in (-\frac{1}{2},\frac{1}{2})$, and $s>\frac{7}{2}+m$, $m\in \N_0$. Let $h$ be the forward solution of the equation $L_bh=f\in C_c^{\infty}((0,\infty)_{t_*};\Hbext^{s,\ell+2})$. Then $h$ can be written as
\begin{align*}
h=\hat{h}+\tilde{h}_t+\tilde{h},
\end{align*}
where $\hat{h}\in \hat{\cK}_b$ is a generalized mode,\footnote{Recall here Proposition~\usref{prop6.4} and Theorem~\usref{ThmL0}.} $\tilde{h}_t\in \la t_*\ra^{-\frac12+\epsilon}H_{\bop,{\rm c}}^{\infty,m}(\R_{t_*};\cK_{b,t})$, and where $\tilde{h}$ satisfies the decay estimate
\begin{align*}
\Vert \tilde{h}\Vert_{\la t_*\ra^{-\frac32+\epsilon}\wt{H}_{\bop,\rm c}^{s-2,\ell+\epsilon-1,m}}\lesssim \Vert f\Vert_{\la t_*\ra^{-\frac72+\epsilon}\wt{H}_{\bop,\rm c}^{s,\ell+2,m}}. 
\end{align*}

If $m\ge 1$, then we also have
\begin{align*}
  &\Vert (\la t_*\ra D_{t_*})^{m-1}\tilde{h}\Vert_{\la t_*\ra^{-2+\epsilon}L^{\infty}(\R_{t_*};\wt{H}_{\bop}^{s-2-m,\ell+\epsilon+1})} \\
  &\quad + \|(\la t_*\ra D_{t_*})^{m-1}\tilde h_t\|_{\la t_*\ra^{-1+\eps}L^\infty(\R_{t_*};\cK_{b,t})}\lesssim \Vert f\Vert_{\la t_*\ra^{-\frac72+\epsilon}\wt{H}_{\bop,\rm c}^{s,\ell+2,m}}.
\end{align*} 
In particular, for $m\geq 1$, $\N_0\ni j<s-(\frac{9}{2}+m)$, and $j'\in\N_0$, we have pointwise decay
\begin{align*}
  \vert (\la t_*\ra D_{t_*})^{m-1}V^j\tilde{h}\vert &\lesssim \la t_*\ra^{-2+\epsilon}r^{-\frac12-\ell-\epsilon}\Vert f\Vert_{\la t_*\ra^{-\frac72+\epsilon}\wt{H}^{s,\ell+2,1}_{\bop,\rm c}}, \\
  |(\la t_*\ra D_{t_*})^{m-1}\pa_{t_*}^{j'} \tilde h_t| &\lesssim \la t_*\ra^{-1+\epsilon}\|f\|_{\la t_*\ra^{\frac72+\eps}\tilde H_{\bop,{\rm c}}^{s,\ell+2,1}},
\end{align*} 
where $V^j$ is any up to $j$-fold decomposition of $\partial_{t_*}$, $r\partial_r$, and rotation vector fields. In the notation of~\eqref{EqRZeroModes} and Definition~\usref{DefRBreve}, the leading order term $\hat h$ can be expressed in terms of $f$ in the following manner: we have
  \[
    \hat h = (t_*h_t + \breve h_t) + h'_t + g_b^{\prime\Ups}(\dot{b}) + h''_t,\quad h_t,\ h'_t,\ h''_t\in\cK_{b,t},\ g_b^{\prime\Ups}(\dot{b}) \in\cK_{b,l},
  \]
  where $h_t,h'_t,g_b^{\prime\Ups}(\dot{b})$ are determined by~\eqref{EqRSing1}--\eqref{EqRSing2} for $\int_{\R_{t_*}} f(t_*)\,d t_*$ in place of $f$, and $h''_t$ is the unique element in $\cK_{b,t}$ for which there exists $g_b^{\prime\Ups}(\dot{b})\in\cK_{b,l}$ with $k_b((h''_t,g_b^{\prime\Ups}(\dot{b}),\cdot)=-\int_{\R_{t_*}}t_* f(t_*)\,d t_*$ as elements of $(\cK_b^*)^*$.
\end{thm}

The estimate for $\tilde h_t$ is from \cite[(0.2b)--(0.2c)]{HaefnerHintzVasyKerrErratum}; the infinite regularity of $\tilde h_t$ with respect to $\pa_{t_*}$ (on the spectral side: rapid decay as $\sigma\to\pm\infty$) follows from the fact that, on the spectral side, $\tilde h_t$ arises from a singular term of the resolvent at $\sigma=0$.

We now give a version for the initial value problem. 
\begin{thm}
\label{thm7.2}
  Let $\alpha\in(0,1)$ and $s>\frac{13}{2}+m$, $m\in\N_0$. Suppose\footnote{Since we use a scattering (asymptotically Euclidean) density to define $L^2$-spaces, the assumptions on $h_0,h_1$ imply pointwise $r^{-1-\alpha}$ and $r^{-2-\alpha}$ decay, respectively.}
  \[
    h_0 \in \Hbext^{s,-\frac12+\alpha}(\Sigma_0;S^2\,{\Tsc^*}\Sigma_0),\quad
    h_1 \in \Hbext^{s-1,\frac12+\alpha}(\Sigma_0;S^2\,{\Tsc^*}\Sigma_0).
  \]
  Then the solution $h$ of the initial value problem
  \[
      L_b h = 0, \qquad
      \bigl(h|_{\Sigma_0},\, (\cL_{\pa_\ft}h)|_{\Sigma_0}\bigr) = (h_0,\, h_1).
  \]
  has the following asymptotic behavior:
  \begin{enumerate}
  \item\label{ItPfIVP1} in $t_*\geq 0$, we can write $h=\hat h+\tilde{h}_t+\tilde h$, where $\hat h\in\wh\cK_b$ is a generalized zero mode of $L_b$, and where the remainders $\tilde h$, $\tilde{h}_t$ satisfy the following decay in $t_*\geq 0$: for $\eps\in(0,1)$ with $\alpha+\eps>1$, we have
    \begin{subequations}
    \begin{align}
    &\tilde{h}_t\in \la t_*\ra^{-\frac12+\epsilon}H_{\bop,{\rm c}}^{\infty,m}(\R_{t_*};\cK_{b,t}),\\
    \label{EqPfIVP1Hb}
     & \|\tilde h\|_{\la t_*\ra^{-\frac32+\eps}\wt H_{\bop,\rm c}^{s-5,-\frac52+\alpha+\eps,m}} \lesssim \|h_0\|_{\Hbext^{s,-\frac12+\alpha}}+\|h_1\|_{\Hbext^{s-1,\frac12+\alpha}}.
    \end{align}
    For $m\geq 1$, we moreover have the $L^\infty$-bounds
    \begin{equation}
    \label{EqPfIVP1Linfty}
    \begin{split}
  &\Vert (\la t_*\ra D_{t_*})^{m-1}\tilde{h}\Vert_{\la t_*\ra^{-2+\epsilon}L^{\infty}(\R_{t_*};\wt{H}_{\bop}^{s-5-m,\ell+\epsilon+1})} \\
  &\quad + \|(\la t_*\ra D_{t_*})^{m-1}\tilde h_t\|_{\la t_*\ra^{-1+\eps}L^\infty(\R_{t_*};\cK_{b,t})}\lesssim \|h_0\|_{\Hbext^{s,-\frac12+\alpha}}+\|h_1\|_{\Hbext^{s-1,\frac12+\alpha}}.
    \end{split}
    \end{equation}
    For $\N_0\ni j<s-(\frac{13}{2}+m)$, we moreover have the pointwise bound
    \begin{equation}
    \label{EqPfIVP1Pt}
      |(\la t_*\ra D_{t_*})^{m-1} V^j \tilde h(t_*,r,\omega)| \lesssim \la t_*\ra^{-1-\alpha+}\left(\frac{\la t_*\ra}{\la t_*\ra+\la r\ra}\right)^{\alpha-},
    \end{equation}
    \end{subequations}
    where $V^j$ is any up to $j$-fold composition of the vector fields $\pa_{t_*}$, $r\pa_r$, and rotation vector fields.
  \item\label{ItPfIVP2}  in $\ft\geq 0$, $t_*\leq 0$, we have $|\tilde h|\lesssim r^{-1}(1+|t_*|)^{-\alpha}$.
  \end{enumerate}
\end{thm}
If one imposes in addition the linearized constraint equations on the data, one obtains:
\begin{thm}
\label{thm7.3}
 Let $\alpha\in (0,1)$, and let $s>\frac{13}{2}+m$, $m\in \N_0$. Suppose that the tensors 
\begin{align*}
\dot{\gamma}\in \Hbext^{s,-\frac12+\alpha}(\Sigma_0,S^2\,\Tsc^*\Sigma_0),\quad \dot{k}\in \Hbext^{s-1,\frac12+\alpha}(\Sigma_0;S^2\,\Tsc^*\Sigma_0)
\end{align*}
satisfy the linearization of the constraint equations around the initial data $(\gamma_b,k_b)$ of a subextremal Kerr metric. Then there exists a solution of the initial value problem 
\begin{align*}
D_{g_b}\Ric(h)=0,\  D_{g_b}\tau(h)=(\dot{\gamma},\dot{k}),
\end{align*}
satisfying the gauge condition $D_{g_b}\Upsilon_b(h;g_b)=-\delta_{g_b}\sfG_{g_b}h=\theta$ which has the asymptotic behavior stated in Theorem~\usref{thm7.1} with $\tilde{h}_t=0$; here $\theta$ is determined by the asymptotic behavior of the solution of a gauge-fixed equation $L_bh=0$ (with Cauchy data determined by $\dot{\gamma}, \dot{k}$) as above. 
\end{thm}
Note that the above theorems slightly improve on the results stated in \cite{HHV21} in the case of small $a$ (due to the present elimination of the Coulomb type zero mode and the fact that the linearization in the mass already shows up as a zero mode). The proofs are based on the precise description of the resolvent. Using the results of \S\ref{Sec9}, they now follow line by line from the arguments in \cite{HHV21,HaefnerHintzVasyKerrErratum}. We thus omit the details.

\appendix
\section{Threshold regularity at the event horizon}

We wish to compute the threshold regularity at the conormal bundle of the event horizon on a subextremal Kerr background $(M_b^\circ,g_b)$ for (the spectral family of) the linearized gauge-fixed Einstein operator (times $2$) (see \cite[Theorem~11.5]{HintzMicro}). Since in this paper we only consider as gauge-fixed operators only small perturbations of the standard wave map/DeTurck version (without constraint damping), the linearization $L$ is a small perturbation of the linearization $\Box_g+2\sR_g$ (i.e.\ the Lichnerowicz wave operator), $g=g_b$, of the standard gauge-fixed version; here $\sR_g$ is a 0-th order term involving the curvature tensor of $g$. Being interested in properties of $L$ near the event horizon, we work over $M_b^\circ$, and in fact in the region $r<3\bhm$ where the metric is given by~\eqref{EqKerrMetric}. We omit the subscript `$b$' from now on.

The subprincipal operator \cite{HintzPsdoInner} of $L$ is equal to
\[
  S_\sub(L)=-i\nabla_{H_G}^{\pi^*S^2 T^*M^\circ},\quad G(z,\zeta)=g_z^{-1}(\zeta,\zeta).
\]
We need to compute its form at $N^*\cH^+=N^*\{r=r_+\}$. We use the coordinates $(t_*,r,\phi_*,\theta)$ as in~\eqref{EqKerrMetric} and write covectors as $-\sigma\,\dd t_*+\xi\,\dd r+\eta_\theta\,\dd\theta+\eta_\phi\,\dd\phi_*$; then one finds
\[
  \varrho^2 G = -2\bigl(-(r^2+a^2)\sigma+a\eta_\phi\bigr)\xi-\Delta\xi^2 - \eta_\theta^2 - \sin^{-2}\theta(\eta_\phi-a\sin^2\theta\,\sigma)^2
\]
and thus
\begin{equation}
\label{EqHamVF}
  \text{at}\ N^*\cH^+:\quad H_{\varrho^2 G} = -2(r_+^2+a^2)\xi\pa_{t_*} - 2 a\xi\pa_{\phi_*} + 2\kappa_+\xi^2\pa_\xi,\quad \kappa_+:=r_+-\bhm.
\end{equation}
On the characteristic set $G^{-1}(0)\supset N^*\cH^+$, this equals $\varrho^2 H_G$. The Hamiltonian vector field of $\varrho^2$ times the principal symbol of an element $\hat L(\sigma)$ of the spectral family at $N^*\{r=r_+\}$ is obtained from this simply by dropping the first term.

We compute the pullback connection on $\pi^*T^*M^\circ$ in the frame
\begin{equation}
\label{EqTstarFrame}
\begin{alignedat}{2}
  e^0 &= \varrho^{-1}(\dd t_*-a\sin^2\theta\,\dd\phi_*), &\qquad
  e^1 &= \varrho\,\dd r, \\
  e^2 &= \varrho\,\dd\theta, &\qquad
  e^3&=\varrho^{-1}\sin\theta(a\,\dd t_*-(r^2+a^2)\,\dd\phi_*)
\end{alignedat}
\end{equation}
in which $g=\Delta(e^0)^2-2 e^0\otimes_s e^1-(e^2)^2-(e^3)^2$. In the splitting
\[
  T^*M^\circ = \la e^0\ra \oplus \la e^1\ra \oplus \la e^2\ra \oplus \la e^3\ra,
\]
a short calculation (at $N^*\cH^+$, and using $\xi^{-1}\pi_*H_{\varrho^2 G}=-2 e_0$ there, where $e_\mu$ is the dual frame to $e^\mu$, and the Koszul formula) then gives
\[
  \nabla_{H_G}^{\pi^*T^*M^\circ} = H_G + \xi\varrho^{-2}S_{(1)},\quad
  S_{(1)} = \begin{pmatrix} a_0 & 0 & 0 & 0 \\ 0 & -a_0 & -a_1 & -a_2 \\ a_1 & 0 & 0 & 0 \\ a_2 & 0 & 0 & 0 \end{pmatrix},
\]
where
\[
  a_0 = 2\kappa_+,\quad
  a_1 = -a^2\varrho^{-2}\sin(2\theta),\quad
  a_2 = 2 a r_+\varrho^{-2}\sin\theta.
\]
In the induced splitting (writing $e^\mu e^\nu:=e^\mu\otimes_s e^\nu$ for brevity)
\begin{equation}
\label{EqS2Split}
\begin{split}
  S^2 T^*M^\circ &= \la e^0 e^0\ra \oplus \la 2 e^0 e^1\ra \oplus \la 2 e^0 e^2\ra \oplus \la 2 e^0 e^3\ra \\
    &\qquad \oplus \la e^1 e^1 \ra \oplus \la 2 e^1 e^2\ra \oplus \la 2 e^1 e^3\ra \oplus \la e^2 e^2\ra \oplus \la 2 e^2 e^3\ra \oplus \la e^3 e^3\ra,
\end{split}
\end{equation}
we therefore have, upon computing the second symmetric tensor power of $S_{(1)}$:
\begin{lemma}[Subprincipal operator for $L$ on symmetric $2$-tensors]
\label{LemmaSubprL}
  At $N^*\cH^+$ and in the splitting~\eqref{EqS2Split}, we have
  \[
    i S_\sub(L) = H_G\otimes 1_{10\times 10} + \xi\varrho^{-2}S_{(2)}
  \]
  where $S_{(2)}$ is the $10\times 10$ matrix
  \[
    S_{(2)}=
    \begin{pmatrix}
      2 a_0 & 0 & 0 & 0 & 0 & 0 & 0 & 0 & 0 & 0 \\
      0 & 0 & -a_1 & -a_2 & 0 & 0 & 0 & 0 & 0 & 0 \\
      a_1 & 0 & a_0 & 0 & 0 & 0 & 0 & 0 & 0 & 0 \\
      a_2 & 0 & 0 & a_0 & 0 & 0 & 0 & 0 & 0 & 0 \\
      0 & 0 & 0 & 0 & -2 a_0 & -2 a_1 & -2 a_2 & 0 & 0 & 0 \\
      0 & a_1 & 0 & 0 & 0 & -a_0 & 0 & -a_1 & -a_2 & 0 \\
      0 & a_2 & 0 & 0 & 0 & 0 & -a_0 & 0 & -a_1 & -a_2 \\
      0 & 0 & 2 a_1 & 0 & 0 & 0 & 0 & 0 & 0 & 0 \\
      0 & 0 & a_2 & a_1 & 0 & 0 & 0 & 0 & 0 & 0 \\
      0 & 0 & 0 & 2 a_2 & 0 & 0 & 0 & 0 & 0 & 0
    \end{pmatrix}.
  \]
  The eigenvalues of $S_{(2)}$ (with multiplicity) are
  \begin{equation}
  \label{EqSEigval}
    0\ (4\times),\quad
    -2 a_0,\quad
    -a_0\ (2\times),\quad
    a_0\ (2\times),\quad
    2 a_0.
  \end{equation}
\end{lemma}
\begin{proof}
  Denote the basis vectors in~\eqref{EqS2Split} by $f_1,\ldots,f_{10}$ (in this order), and write $F:=\la f_1,\ldots,f_{10}\ra$ and $F_j=\la f_j\ra$. Now, $S_{(2)}$ preserves $F_5$, on which it is scalar multiplication by $-2 a_0$. On the quotient $F':=F/F_5$, $S_{(2)}$ preserves (the image of) $F_6\oplus F_7$, on which it is scalar multiplication by $-a_0$. On $F''=F'/(F_6\oplus F_7)$, $S_{(2)}$ annihilates the image of $F_2\oplus F_8\oplus F_9\oplus F_{10}$, and then on $F'''=F''/(F_2\oplus F_8\oplus F_9\oplus F_{10})$, $S_{(2)}$ is lower triangular with diagonal entries $2 a_0$ (on the image of $F_1$) and $a_0$ (twice, on the images of $F_3$ and $F_4$). In other words, in the block splitting $T_1 \oplus T_2 \oplus T_3 \oplus T_4 \oplus T_5$ where
  \begin{equation}
  \label{EqT12345}
  \begin{split}
    T_1&=\la e^1 e^1\ra, \\
    T_2&=\la 2 e^1 e^2\ra\oplus \la 2 e^1 e^3\ra, \\
    T_3&=\la 2 e^0 e^1\ra \oplus \la e^2 e^2\ra \oplus \la 2 e^2 e^3\ra \oplus \la e^3 e^3\ra, \\
    T_4&=\la e^0 e^0\ra, \\
    T_5&=\la 2 e^0 e^2\ra \oplus \la 2 e^0 e^3\ra,
  \end{split}
  \end{equation}
  the matrix of $S_{(2)}$ is lower block-triangular, with its diagonal entries being themselves lower triangular (with diagonal parts $-2 a_0$, $-a_0\cdot 1_{2\times 2}$, $0\cdot 1_{4\times 4}$, $2 a_0$, $a_0\cdot 1_{2\times 2}$).
\end{proof}

We can now compute the threshold regularity; we do this at $N^*\cH^+\cap\{\xi<0\}$ (which lies in the future characteristic set, cf.\ the sign of $\pa_{t_*}$ in~\eqref{EqHamVF}). Set $\rho:=(-\xi)^{-1}$; thus $\beta_0:=\rho^{-1}\cdot\rho\varrho^2 H_G(\rho)=2\kappa_+$ at $N^*\cH^+$. As a consequence of Lemma~\ref{LemmaSubprL}, one can pick an inner product on $\pi^*S^2 T^*M^\circ$, which we can moreover take to be $\phi_*$-invariant, with respect to which
\[
  \varrho^2 \rho\upsigma^1\Bigl(\frac{1}{2 i}(L-L^*)\Bigr) = \varrho^2 (-\xi)^{-1}\frac{1}{2 i}\bigl(S_\sub(L)-S_\sub(L^*)\bigr) < 2 a_0 = 4\kappa_+ + \eps
\]
for any fixed $\eps>0$; this uses \cite[Proposition~3.11(2)]{HintzPsdoInner} and a simple variant of \cite[\S{3.4}]{HintzPsdoInner}. Furthermore, one uses that the spaces $T_1,\ldots,T_5$ in~\eqref{EqT12345} (although not all of the individual 1-dimensional summands) extend to give a smooth splitting of $S^2 T^*M^\circ$ also near the poles $\theta=0,\pi$ of the 2-sphere; this is due to the smoothness of $\la e^2\ra\oplus\la e^3\ra$ as a rank $2$ subbundle of $T^*M^\circ$ and thus of $\pi^*T^*M^\circ$ (with the fiber over a point $(t_*,r,\omega)$, $\omega\in\Sph^2$, given by $\{ \eta+\frac{a}{r^2+a^2}\eta(\pa_{\phi_*})\dd t_*\colon \eta\in T^*_\omega\Sph^2\}$). The threshold regularity is then (see also \cite[Proposition~4.9]{HintzGlueLocII})
\[
  \frac{1}{2} + \frac{4\kappa_++\eps}{\beta_0} = \frac{5}{2} + \frac{\eps}{2\kappa_+}.
\]
Since $\eps>0$ is arbitrary, we conclude that the threshold regularity for the radial source estimate for the operator $L$ at $N^*\cH^+\cap\{\pm\xi<0\}$ is, indeed, $\frac52$.

\bibliographystyle{alpha}

\end{document}